\documentclass[11pt]{article}

\usepackage{amsmath}
\usepackage{amsthm}
\usepackage{amssymb}
\usepackage{algorithm}
\usepackage{color}
\usepackage[english]{babel}
\usepackage{graphicx}
\usepackage{grffile}
\usepackage{wrapfig,epsfig}
\usepackage{epstopdf}
\usepackage{url}
\usepackage{color}
\usepackage{algpseudocode}
\usepackage[T1]{fontenc}
\usepackage{bbm}
\usepackage{comment}
\usepackage{dsfont}
\usepackage{thm-restate}
\usepackage{bm}
\usepackage{tcolorbox}
\usepackage{subcaption}

\usepackage{enumitem}

\usepackage[margin=1in]{geometry}

\graphicspath{{./figs/}}
\usepackage{mathtools}

\newtheorem{theorem}{Theorem}[section]
 
\newtheorem{lemma}[theorem]{Lemma}
\newtheorem{definition}[theorem]{Definition}

\newtheorem{remark}[theorem]{Remark}
\newtheorem*{remark*}{Remark}
\newtheorem{claim}[theorem]{Claim}

\newcommand{\wt}{\widetilde}
\newcommand{\eps}{\epsilon}

\newcommand{\R}{\mathbb{R}}

\newcommand{\bx}{\mathbf{x}}
\newcommand{\ba}{\mathbf{a}}
\newcommand{\by}{\mathbf{y}}

\newcommand{\bb}{\mathbf{b}}

\newcommand{\bg}{\mathbf{g}}
\newcommand{\bu}{\mathbf{u}}
\newcommand{\br}{\mathbf{r}}

\newcommand{\bv}{\mathbf{v}}

\newcommand{\mB}{\mathbb{B}}
\newcommand{\mA}{\mathcal{A}}
\newcommand{\close}{\mathsf{close}}
\newcommand{\T}{\mathsf{T}}
\newcommand{\mO}{\mathsf{O}}
\newcommand{\mR}{\mathsf{R}}
\newcommand{\mS}{\mathbb{S}}
\newcommand{\calS}{\mathcal{S}}
\newcommand{\mJ}{\mathcal{J}}
\newcommand{\bq}{\mathbf{q}}

\newcommand{\Binghui}[1]{{\color{blue}[Binghui: #1]}}

\newcommand{\bX}{\mathbf{X}}

\newcommand{\bR}{\mathbf{R}}

\newcommand{\bM}{\mathbf{M}}
\newcommand{\bA}{\mathbf{A}}
\newcommand{\bB}{\mathbf{B}}

\newcommand{\nil}{\mathsf{nil}}
\newcommand{\rank}{\mathsf{rank}}

\newcommand{\nA}{\mathcal{A}_{\mathsf{nice}}}
\newcommand{\row}{\mathsf{row}}

\newcommand{\bI}{\mathbf{I}}

\newcommand{\mH}{\mathcal{H}}

\newcommand{\bU}{\mathbf{U}}

\newcommand{\mE}{\mathcal{E}}

\newcommand{\mN}{\mathcal{N}}

\newcommand{\mI}{\mathcal{I}}

\newcommand{\suc}{\mathsf{suc}}

\newcommand{\bS}{\mathbb{S}}
\newcommand{\sP}{\mathsf{P}}

\DeclareMathOperator*{\E}{{\mathbb{E}}}

\DeclareMathOperator{\ALG}{ALG}

\DeclareMathOperator{\poly}{poly}

\DeclareMathOperator{\argmin}{argmin}

\DeclareMathOperator{\proj}{proj}

\usepackage{braket}

\algnewcommand\algorithmicforeach{\textbf{for each}}
\algdef{S}[FOR]{ForEach}[1]{\algorithmicforeach\ #1\ \algorithmicdo}

\makeatletter
\newcommand*{\RN}[1]{\expandafter\@slowromancap\romannumeral #1@}
\makeatother

\title{Memory-Query Tradeoffs for Randomized Convex Optimization}
\author{}
\author{  
Xi Chen \\ Columbia University \\ texttt{xichen@cs.columbia.edu}
\and Binghui Peng \\ Columbia University \\ \texttt{bp2601@columbia.edu}
}
\date{}

\begin{document}
\maketitle

\begin{abstract}
We show that any randomized first-order algorithm which minimizes a
$d$-dimensional, $1$-Lipschitz convex function over the unit ball must either use $\Omega(d^{2-\delta})$ bits of memory or make $\Omega(d^{1+\delta/6-o(1)})$ queries, for any constant $\delta\in (0,1)$ and when the precision $\epsilon$ is quasipolynomially small in $d$. Our result implies that cutting plane methods, which use $\tilde{O}(d^2)$ bits of memory and $\tilde{O}(d)$ queries, are Pareto-optimal among randomized first-order algorithms, and quadratic memory is required to achieve optimal query complexity for convex optimization.
\end{abstract}

\setcounter{page}{0}
\thispagestyle{empty}
\newpage

\section{Introduction}

A fundamental problem in optimization and mathematical programming is convex optimization given access to a first-order oracle.
Consider one of the canonical settings where the input is a $1$-Lipschitz, convex function $F:\mB^d\rightarrow \mathbb{R}$ over the $d$-dimensional unit ball $\mB^d$. An algorithm has access to a first-order oracle of $F$: In each round, it can send a query point $\bx\in \mB^d$
to the oracle and receive a pair $(F(\bx),\bg(\bx))$, where $\bg(\bx)\in \partial F(\bx)\subseteq \mathbb{R}^d$ is a subgradient of $F$ at $\bx$.
The goal is to find an $\epsilon$-optimal point $\bx^*\in \mB^d$ satisfying 
$F(\bx^*)-\argmin_{\bx\in \mB^d}F(\bx)\le \epsilon$.
Convex optimization has a wide range of applications in numerous fields. In particular, it has served as one of the most important primitives in machine learning \cite{bubeck2015convex}

\def\eps{\epsilon}

The worst-case query complexity of minimizing a $1$-Lipschitz, convex function $F:\mB^d\rightarrow \mathbb{R}$  has long been known to be $\Theta(\min\{1/\eps^2,d\log (1/\epsilon)\})$ \cite{nemirovskij1983problem}.
The two upper bounds can be achieved by subgradient descent ($O(1/\eps^2)$ \cite{nesterov2003introductory}) and cutting-plane methods ($O(d\log (1/\eps))$ \cite{levin1965algorithm,bertsimas2004solving,vaidya1996new,atkinson1995cutting,lee2015faster,jiang2020improved}), respectively. 
However, even when $\eps\ll 1/\sqrt{d}$ (which is arguably the more interesting regime), gradient descent has been the dominant approach for convex optimization despite its  suboptimal worst-case query complexity, and cutting-plane methods are less frequently used in practice.
One noticeable drawback of cutting-plane methods is its memory requirement, which has become a more and more important resource in the era of massive datasets. 
While cutting-plane methods require quadratic $\Omega(d^2\log (1/\eps))$ bits of memory (e.g., to either store all subgradients queried so far or at least an ellipsoid in $\mathbb{R}^d$),   in contrast, linear $O(d\log (1/\eps))$ bits of memory suffice for gradient descent. 
It is a natural question to ask whether there is an algorithm that can achieve the better  of the two worlds.

This motivates a COLT 2019 open problem posed by  Woodworth and  Srebro \cite{WSOpen} to study \emph{memory-query tradeoffs of convex optimization}.
The first result along this direction was proved by Marsden, Sharan, Sidford and Valiant \cite{marsden2022efficient}. They showed that any randomized first-order algorithm that minimizes a $1$-Lipschitz, convex function $F:\mB^d\rightarrow \mathbb{R}$ with $\eps=1/\text{poly}(d)$ accuracy must either use $d^{1.25-\delta}$ bits of memory or make $d^{1+(4/3)\delta}$ many queries, for any constant $\delta\in [0,1/4]$.
Recently, Blanchard, Zhang and Jaillet \cite{blanchard2023quadratic} showed that any {\em deterministic} algorithm must use either $d^{2-\delta}$ bits of memory or $d^{1+\delta/3}$ queries.
Their tradeoffs imply that cutting-plane methods are Pareto-optimal among  deterministic algorithms
on the memory-query curve of  convex optimization.

\subsection{Our Contribution}

The result of Blanchard, Zhang and Jaillet \cite{blanchard2023quadratic} left open the question about whether similar tradeoffs hold for randomized  algorithms.
We believe that this is an important question to address because  numerous  gradient descent methods (e.g. stochastic gradient descent \cite{bubeck2015convex}, randomized smoothing \cite{duchi2012randomized}, variance reduction \cite{johnson2013accelerating}) are inherently randomized. Many exponential separations are known between deterministic and randomized algorithms for optimization problems, including escaping saddle points \cite{du2017gradient, jin2021nonconvex} and volume estimations \cite{vempala2005geometric}. The role of randomness becomes even more critical in connection to memory. Almost all known streaming algorithms and dimension reduction technique are randomized \cite{woodruff2014sketching}.


In this paper, we show that cutting-plane methods are in fact Pareto-optimal among randomized algorithms and thus, \emph{optimal query complexity 
  for convex optimization requires 
  quadratic memory.}
This is a corollary of the following memory-query tradeoff:

\begin{theorem}\label{theo:main}
Let $d$ be a sufficiently large integer, $\eps = \exp(-\log^5 d)$ and $\delta \in (0,1)$. 
Any algorithm that finds an $\eps$-optimal point of a $d$-dimensional $1$-Lipschitz convex function  requires either $\Omega(d^{2-\delta})$ bits of memory or makes at least $\Omega(d^{1+\delta/6-o(1)})$ queries.
\end{theorem}



After reviewing additional related work in Section \ref{sec:related}, we give an overview of our techniques in Section \ref{sec:overview}. At a high level, our proof is based on a novel reduction to a two-player, three-round communication problem which we call the \emph{correlated orthogonal vector game}.
Crucially, it differs from the \emph{orthogonal vector game} considered in both \cite{marsden2022efficient} and \cite{blanchard2023quadratic} in  two aspects: (a) It suffices for Bob to output a single vector that is orthogonal to vectors of Alice, but (b) the vector cannot be an arbitrary one but needs to have a strong correlation with a random vector that Bob sends to Alice in the second round. See Section \ref{sec:sketchgame}.
The major challenge is to prove a lower bound for the correlated orthogonal vector game, from which Theorem \ref{theo:main} follows.
For this purpose, we develop a recursive encoding scheme which we explain in more details in Section \ref{sec:sketchlowerbound}.
We believe that our techniques may have applications in understanding memory-query tradeoffs in other settings.


\subsection{Additional related work}\label{sec:related}
\paragraph{Learning with limited memory} The role of memory for learning has been extensively studied in the past few years, including memory-sample tradeoffs for parity learning \cite{raz2017time,raz2018fast,garg2018extractor,garg2019time,garg2021memory,liu2023memory} and linear regression \cite{sharan2019memory}, memory-regret tradeoffs for online learning \cite{peng2022online, peng2023near, woodruff2022memory,woodruff2023streaming,aamand2023improved}, memory bounds for continual learning \cite{chen2022memory}, communication/memory bounds for statistical estimation \cite{balcan2012distributed, garg2014communication,steinhardt2015minimax,braverman2016communication} and many others \cite{steinhardt2016memory, dagan2018detecting,dagan2019space,gonen2020towards, braverman2020gradient, feldman2020does, brown2021memorization, brown2022strong}.

\paragraph{Literature of convex optimization} There is a long line of literature on convex optimization, we refer readers to \cite{nesterov2003introductory, bubeck2015convex} for a general coverage of the area. 
In particular, there is a long line of work \cite{nemirovskij1983problem, nesterov2003introductory,woodworth2016tight,braun2017lower,simchowitz2018tight, braverman2020gradient} on establishing query lower bounds for finding approximate minimizers of Lipschitz functions, with access to a subgradient oracle. A recent line of work \cite{nemirovski1994parallel,balkanski2018parallelization, bubeck2019complexity, diakonikolas2019lower} studies lower bounds for parallel algorithms.
Memory efficient optimization algorithms have also been proposed, including limited-memory-BFGS \cite{liu1989limited,nocedal1980updating}, conjugate gradient descent \cite{hestenes1952methods,fletcher1964function,hager2006survey}, sketched/subsampled quasi-newton method \cite{pilanci2017newton,roosta2019sub} and recursive cutting-planes \cite{blanchard2023constrained}.

\section{Technique overview}\label{sec:overview}

We provide a high-level overview of   techniques behind the proof of Theorem \ref{theo:main}. To establish the memory-query tradeoffs, we use the same family of hard instance as \cite{marsden2022efficient} and our contribution is a novel and improved analysis. Define  
\begin{align}
F(\bx) =  \frac{1}{\sqrt{d}L} \cdot \max \left\{L \|\bA\bx\|_\infty - 1, \hspace{0.06cm}\max_{i \in [N]}\Big(\langle \bv_i, \bx\rangle - i \gamma\Big) \right\} \label{eq:obj-over}
\end{align}
where $\smash{\bA \sim \{-1,1\}^{(d/2)\times d}}$, $\smash{\bv_i \sim \frac{1}{\sqrt{d}}\mH_d}$ independently and uniformly at random with $\mH_d:= \{-1,1\}^d$.

Letting $\delta\in (0,1)$ be the constant in Theorem \ref{theo:main}, we use the standard choice of $$\gamma = \wt{O}(d^{-\delta/4})\quad\text{and}\quad N = \wt{O}(d^{\delta/6})$$ for Nemirovski function, and $L = \exp(\log^5 d )$ is a large scaling factor.

The function \eqref{eq:obj-over} consists of two parts: the projection term $\|\bA\bx\|_\infty$ and the Nemirovski function $\max_{i \in [N]}(\langle \bv_i, \bx\rangle - i \gamma)$. From a high level, the Nemirovski function enforces that any convex optimization algorithm, with high probability, has to obtain $\bv_1, \ldots, \bv_{N}$  in this order, and  to obtain the subgradient $\bv_{i+1}$, the algorithm must  submit a query $\bx$ that has $\Omega(\gamma)$-correlation with $\bv_i$ and at the same time,
  is almost orthogonal to $\bA$, as enforced by the projection term. 


\subsection{Reduction to correlated orthogonal vector game}\label{sec:sketchgame}

At the heart of our proof is the new correlated orthogonal vector game. It is an abstract model that captures the hardness of the Nemirovski function and the projection in the hard functions $F$.
\begin{definition}[$(s, \xi)$-Correlated orthogonal vector game]
Let $d\in \mathbb{N}$ be the dimension. Let $k, s, n \in [d]$ and $\xi \in (0, 1)$ be input parameters.
The $(s, \xi)$-correlated orthogonal vector game is a two-player three-round communication problem: Alice receives a matrix $\bA \sim \{-1,1\}^{(d/2)\times d}$ and Bob receives a vector $\bv \sim \frac{1}{\sqrt{d}}\mH_d$. A deterministic $(k, n)$-communication protocol proceeds in three rounds:
\begin{enumerate}
\item[] Round 1: Alice sends Bob a message $\bM$ of length $dk$.
\item[] Round 2: Bob sends Alice the vector $\bv$.
\item[] Round 3: Alice sends Bob $n$ rows $\ba_1, \ldots, \ba_n \in \row(\bA) \cup \{\nil\}$. 
\end{enumerate}
At the end, Bob is required to output a vector $\bx \in \mathbb{B}^d$, such that, with probability at least $1/2$ (over the randomness of $\bA$ and $\bv$), 
\begin{itemize}
\item []  {\bf Orthogonality:} $\bx$ is almost orthogonal to $\bA$, i.e., $\|\bA\bx\|_{\infty} \leq \xi$, and
\item []  {\bf Correlation:} $\bx$ has large correlations with $\bv$, i.e., $|\langle \bv, \bx \rangle| \geq \sqrt{s/d}$.
\end{itemize}
A randomized $(k,n)$-communication protocol is a distribution of deterministic protocols.
\end{definition}

Suppose that there is a randomized convex optimization algorithm that uses $S=d^{2-\delta}$ bits of memory, makes $T=d^{1+\delta/6-o(1)}$ queries, and finds an $\eps$-optimal solution of a random $F$  with probability at least $2/3$, where $\eps$ is quasipolynomially small in $d$.
We show in Section \ref{sec:reduction} that such an algorithm can be used to obtain a $(k,n)$-communication protocol with $k=S/d$ and $n\approx T/N$ that solves the $(s,\xi)$-correlated orthogonal vector game with $s\approx d\gamma^2\approx d^{1-\delta/2} $ and a quasipolynomially small $\xi$. The novelty is mainly in the definition of the correlated orthogonal vector game and that we can manage to prove a strong lower bound for it, as sketched in the next subsection; the reduction itself uses standard techniques from the literature.

\subsection{Lower bound for the correlated orthogonal vector game}\label{sec:sketchlowerbound}
We next that prove any $(k,n)$-communication protocol for $(s, \xi)$-correlated orthogonal vector game, requires either sending $kd \geq \Omega(s^{2})$ bits in the first round or sending $n \geq d^{1-o(1)}$ rows in the third round.
By our choice of parameters, this implies the third round message must contain $T/N = d^{1-o(1)}$ rows, and  yields a query lower bound of  $T\geq d^{1+\delta/6 - o(1)}$.


\begin{remark}[Sharpness of parameters] We remark on the relation of $k,s,d, n$.
A communication protocol that sends $n = d/2$ rows in the third round could of course resolve the game. Hence $n \geq d^{1-o(1)}$ is almost the best one can prove.
Meanwhile, for the first round, Alice could send a $k$-dimensional subspace $U$ (this costs roughly $kd$ bits) that is orthogonal to $\bA$, and even without the third round message, Bob can extract a vector $\bx \in U$ that satisfies $|\langle \bx, \bv\rangle| \geq O(\sqrt{k/d})$. This implies that $k \geq \Omega(s)$ is the best one can prove. 
Our lower bound does not exactly match this and we only manage to prove $kd \geq \Omega(s^2)$. Though the later already suffices for an improved (quadratic) lower bound for randomized algorithms.
\end{remark}

The proof of this lower bound turns out to be  challenging and it is our main technical contribution. We prove by contradiction and assume (1) $kd \ll s^2$ and (2) $ \Delta := d/n\geq \exp(\log d/\log\log d)$.
Our lower bound is established via an iterative encoding argument: We provide an encoding procedure, built upon the (too-good-to-be-true) communication protocol, that encodes (almost) all matrices $\bA \in \{0, 1\}^{(d/2)\times d}$ into a set of size no more than $2^{d^2/2}/d$.

\paragraph{Warm up: Encoding in the bit model} Our encoding argument is delicate and we first illustrate the basic idea in a simpler ``bit model'' of the correlated orthogonal vector game:  Alice sends $n$ bits (instead rows) in the third round. This is strictly weaker than the original model because $n$ rows can transmit at least $n$ bits of message (even when they are required to be row vectors of $\bA$). 

In this model, Bob's output vector $\bx$, is a function of the first round message $\bM \in \{0,1\}^{kd}$, the third round message $\bb \in \{0,1\}^{n}$ and its input vector $\bv \in \frac{1}{\sqrt{d}}\mH_{d}$, and we denote as $\bx_{\bM, \bv, \bb}$. Fix a matrix $\bA \in \{-1,1\}^{(d/2)\times d}$ and the first round message $\bM_\bA$, the choice of $\bv$ is still uniformly at random, and we observe that the collection $\{\bx_{\bM_\bA, \bv, \bb}\}_{\bv \in \frac{1}{\sqrt{d}}\mH_{d}, \bb\in \{0,1\}^n}$ should contain at least $s$ linearly independent vectors that are orthogonal to $\bA$, and this continues to hold w.h.p. when one restricts to a subset $V^{*} \subseteq \mH_{d}$ of size $O(s)$. 
That is to say, using a probabilistic argument, there exists a set $V^{*} \subseteq \frac{1}{\sqrt{d}}\mH_d$ ($|V^{*}| = O(s)$) such that $\{\bx_{\bM_\bA, \bv, \bb}\}_{\bv \in V^{*}, \bb \in \{0,1\}^n}$ contains at least $s$ linearly independent vectors that are orthogonal to $\bA$, for at least $1/2$-fraction of matrices $\bA \in \{0, 1\}^{(d/2)\times d}$. We denote this set as $\nA$

Consider the following (one-shot) encoding strategy: Fix a message $\bM \in \{0,1\}^{kd}$, a $s$-tuple $\bx_1, \ldots, \bx_{s} \in \{\bx_{\bM, \bv, \bb}\}_{\bv \in V^{*}, \bb \in \{0,1\}^n}$ that are linearly independent, include all matrices $\bA \in \{-1,1\}^{(d/2)\times d}$ whose rows are orthogonal to $\bx_{1}, \ldots, \bx_{s}$.
The above argument implies that all matrices $\bA \in \nA$ are included, but on the hand, the encoding scheme only encodes 
\[
2^{kd} \times (2^n s)^{s} \times 2^{(d/2)\times (d - s)} \leq 2^{d^2/2+ kd + ns\log_2(s) -ds/2} \ll \frac{1}{2} \cdot 2^{d^2/2} = |\nA| 
\]
different matrices. Here the first term of LHS is the number of message $\bM$, the second term is the number of $s$-tuple $\bx_1, \ldots, \bx_{s}$ (note this is the place we need $|V^{*}| \leq O(s)$) and the third term is the number of matrices $\bA$ that are orthogonal to the $s$-tuple.

\paragraph{Encoding in the row model} We next design an encoding strategy based on the (too-good-to-be-true) communication protocol that sends $n$ rows instead of $n$ bits, which turns out to be much more challenging and requires an iterative encoding.
In the row model, the output $\bx$ of Bob depends on the message $\bM \in \{0,1\}^{kd}$, the vector $\bv \in \frac{1}{\sqrt{d}}\mH_{d}$ as well as $n$ rows $\bR \subseteq \row(\bA) \cup \{\nil\}$, and we denote it as $\bx_{\bM, \bv, \bR}$. 
The first step is similar and we select a subset $V^{*} \subseteq \frac{1}{\sqrt{d}}\mH_d$ that is ``representive''.
In particular, one can select a subset $V^{*} \subseteq \frac{1}{\sqrt{d}}\mH_{d}$, of size at most $2^{s}$, such that $\{\bx_{\bM_\bA, \bv, \bR}\}_{\bv \in V^{*}, \bR \in \row(\bA) \cup \{\nil\}}$ contains at least $N_0 = 2^{\Omega(s)}$ ``well-spread'' output vectors $\bx_1, \ldots, \bx_{N_0}$ that are orthogonal to $\bA$, for $1/2$-fraction of matrices $\bA \in \{-1,1\}^{(d/2)\times d}$. 
Here ``well-spread'' means the pairwise distance are at least $\alpha_{0} = 1/d^8$. Again, we denote this set as $\nA$.



To apply a similar encoding strategy, we need to fix the message $\bM$, and the $s$-tuple from $\T^{\bM}(\bA) := \{\bx_{\bM, \bv, \bR}\}_{\bv \in V^{*}, \bR \in \row(\bA)\cup \{\nil\}}$, this turns out to be circular because it already needs the full knowledge of $\bA$. 
Hence, a natural idea is to select a subset of rows $\bA_{1}$ of $\bA$, and uses the $s$-tuple from $\T^{\bM}(\bA_{1}) = \{\bx_{\bM, \bv, \bR}\}_{\bv \in V^{*}, \bR \in \row(\bA_{1}) \cup \{\nil\}}$. 
Let $\bA_{2} = \bA\backslash \bA_{1}$ be the removed sub-matrix and we need its size to be as large as possible -- this is the part where we get compressions.
Meanwhile, one can not hope to remove too many rows, because removing one row would decrease the size of $\T^{\bM}(\bA)$ by a factor of roughly $(1-1/\Delta)$. There are $N_0 = 2^{\Omega(s)}$ well-spread vectors in $\T^{\bM_\bA}(\bA)$ (that are also orthogonal to $\bA$), so one only affords to remove at most $s\Delta$ rows.

The encoding strategy would fix a message $\bM\in \{0,1\}^{kd}$, a submatrix $\bA_1 \in \{-1,1\}^{(d/2- s\Delta)\times d}$, a $s$-tuple $\bx_1, \ldots, \bx_{s} \in \{\bx_{\bM, \bv, \bb}\}_{\bv \in V^{*}, \bR \in \row(\bA_{1})}$ that are (1) linearly independent, and (2) orthogonal to $\bA_{1}$, then it includes all matrices $\wt{\bA}_{2} \in \{-1,1\}^{s\Delta\times d}$ whose rows are orthogonal to $\bx_{1}, \ldots, \bx_{s}$. It encompass $\nA$ and has size at most\footnote{We also need to fix the row indices of $\bA_1$, but it is a lower order term that we omit here for simplicity.}
\[
2^{kd} \times 2^{(d/2-s\Delta)\times d} \times (d^{n}2^s)^s \times 2^{s\Delta\times (d - s)}  \approx 2^{d^2/2 +kd +sn\log_2 (d) -s^2\Delta }.
\]
This is much larger than $2^{d^2/2}$ and not useful at all!
We note in the above expression, the first term is the number of message $\bM$, the second term is the total number of $\bA_{1}$, the third term is the number of $s$-tuple $\bx_1, \ldots, \bx_{s}$ and the last term is the number of orthogonal matrices $\wt{\bA}_{2}$. 


\paragraph{Iterative encoding} Our key observation is that: If there are too many $s$-tuples that are linearly independent and orthogonal to $\bA_{1}$, then there must exists a large number of well-spread vectors in $\T^{\bM_{\bA}}(\bA_{1}) = \{\bx_{\bM_{\bA}, \bv, \bR}\}_{\bv \in V^{*}, \bR \in \row(\bA_{1}) \cup \{\nil\}}$ that are orthogonal to $\bA_{1}$. In particular, if there are $2^{\Omega(s^2\Delta)}$ different $s$-tuples, then there are $N_1:= 2^{\Omega(s\Delta)}$ well-spread vectors that are orthogonal to $\bA_{1}$. This guarantee is stronger than the original one: $N_0 =2^{\Omega(s)}$ well-spread vectors that are orthogonal to $\bA$. Therefore, if we fail to compress $\bA$ by removing $s\Delta$ rows, we can try to compress it by removing another $s\Delta^2$ rows. 
We can repeat it for $\log\log (d)$ iterations and it yields our final iterative encoding scheme.

There are a few subtle yet critical details that are (intentionally) omitted by us, e.g., the ``well-spreadness'' (the choice of pairwise distance $\alpha$) are doubly exponentially decreasing at each iteration, and we need robustly linear independence (similar to \cite{marsden2022efficient}), we refer readers to Section \ref{sec:lower} for these technical details.

\subsection{Comparison with \cite{marsden2022efficient, blanchard2023quadratic}}
We compare our techniques with previous work of \cite{marsden2022efficient, blanchard2023quadratic}. 
Both proofs are based on reductions to the following \emph{orthogonal vector game}:
Alice receives a random matrix $\bA\in \{-1, 1\}^{(d/2)\times d}$, sends Bob a message $\bM \in \{0, 1\}^{kd}$ and $n$ rows $\bR$ from $\bA$, and Bob is required to output $\Omega(k)$ linearly independent vectors that are orthogonal to $\bA$.
On the one hand, \cite{marsden2022efficient} proves a lower bound of $n \geq \Omega(d)$ for orthogonal vector game. 
On the other hand, they show that any $S$-bit, $T$-query randomized algorithm can be used to obtain a protocol for the orthogonal game with $k=S/d$ and $n=T/(N/k)$ and thus, using $n\ge \Omega(d)$, one has
$T=\Omega(Nd^2/S)$. 
However, $N$ can only be $O(d^{1/3})$ in the Nemirovski function, which limits their tradeoffs to apply only up to $S=d^{4/3}$ (for $T$ to be superlinear).
The tradeoffs of \cite{blanchard2023quadratic}, on the other hand, are 
based on the idea that, when against  deterministic algorithms, one can adaptively build  the vector $\bv_{i}$ in the Nemirovski  function (basing on the algorithm's queries so far) and make it orthogonal to previous queries. This increases the size $N$ of Nemirovski function to be $d$ and thus, the tradeoff holds even near $S=d^2$. 
However, 
the argument of \cite{blanchard2023quadratic} is dedicated to deterministic algorithms and can be easily sidestepped by  randomized algorithms.

The correlated orthogonal vector game introduced in this paper is more fine-grained. It proceeds in three rounds (instead of one round as in the orthogonal vector game) and Bob is required to output only one orthogonal vector that has large correlations with the random vector $\bv$.
It is not surprising that our encoding argument gives an alternative proof for the orthogonal vector game: One simply enumerates the message $\bM$, the $n$ rows $\bR$, and the rest rows of matrix $\bA$, which are required to be orthogonal to the $\Omega(k)$ output vectors from $(\bM, \bR)$.
Proving a strong lower bound for the correlated orthogonal vector game is technically the most challenging part of the paper.

\section{Preliminaries}

\paragraph{Notation} Let $[n] = \{1,2,\ldots, n\}$ and $[n_1:n_2] = \{n_1, n_1+1, \ldots, n_2\}$. Let $\mB^{d}$ be the $d$-dimensional unit ball, $\mS^d$ be the $d$-dimensional unit sphere, and $\mH_d = \{-1,1\}^{d}$ be the $d$-dimensional Boolean hypercube.
Given a point $\bx \in \R^d$ and $\delta > 0$, we write $\mB(\bx, \delta):= \{\by: \|\bx -\by\|_2\leq \delta\}$ to denote the $\delta$-ball centered at $\bx$.
For a set of vectors $\bx_1, \ldots, \bx_{m} \in \R^d$, we use $\mathsf{span}(\bx_1, \ldots, \bx_{m}) \subseteq \R^{d}$ to denote the subspace spanned by $\bx_1, \ldots, \bx_{m}$.
Given a subspace $S$ and vector $\bx\in \R^{d}$, we use $\proj_{S}(\bx) \in \R^{d}$ to denote the projection of $\bx$ onto $S$.
Given a matrix $\bA$, we write $\row(\bA)$ to denote the set of its row vectors.
We write $x \sim X$ if the random variable $x$ is drawn uniformly at random from a set $X$.

We consider the following oracle model of convex optimization.
\begin{definition}[Memory constrained convex optimization]
Let $d \in \mathbb{N} $ be the dimension.
An~\mbox{$S$-bit}   and $T$-query convex optimization deterministic algorithm with first-order oracle access runs as follows.
Given access to a first-order oracle of a $1$-Lipschitz convex function $F:$ $ \mathbb{B}^d \rightarrow \R$,
\begin{flushleft}\begin{enumerate}
\item The algorithm starts by initializing  
  the memory to be a string $\bM_0\in \{0,1\}^S$;
\item During iteration $t\in [T]$, the algorithm  
  picks a query $\bx_t\in \mB^d$ based on  $\bM_{t-1}$,
  obtains from the oracle both $F(\bx_t)\in \mathbb{R}$
  and a subgradient $\bg(\bx_t) \in \partial F(\bx_t) \subseteq \R^{d}$, and updates the memory state to $\bM_{t }\in \{0,1\}^S$ based on $ \bM_{t-1},F(\bx_t)$ and $\bg(\bx_t)$.
\item Finally, after $T$ iterations, the algorithm   
  outputs a point $\bx_{T} \in \mathbb{B}^d$ based on $\bM_{T-1}$.
\end{enumerate}\end{flushleft}
An $S$-bit and $T$-query randomized algorithm is a distribution over deterministic algorithms.
We say a randomized algorithm finds an $\epsilon$-optimal point with probability at least $1-\rho$ if for any $f$,  
\[
F(\bx_{T}) - \argmin_{\bx \in \mathbb{B}^d} F(\bx ) \leq \eps.
\]
with probability at least $1-\rho$.
\end{definition}

\section{Reduction to correlated orthogonal vector game}
\label{sec:reduction}

Let $\delta \in (0,1)$ be the constant in Theorem \ref{theo:main}. We use the following hard instances from \cite{marsden2022efficient}:
\[
F(\bx) =  \frac{1}{\sqrt{d}L} \cdot \max \left\{L \|\bA\bx\|_\infty - 1, \hspace{0.06cm}\max_{i \in [N]}\Big(\langle \bv_i, \bx\rangle - i \gamma\Big) \right\}
\]
where $\bA \sim \{-1,1\}^{(d/2) \times d}$, $\bv_i \sim \frac{1}{\sqrt{d}}\mH_d$ (so that each $\bv_i$ is a unit vector), and 
\[
\gamma = \frac{\log^2 d}{d^{\delta/4}}, \quad N = \frac{d^{\delta/6}}{\log^4d} \quad \text{and} \quad L = \exp\big(\log^5d\big).
\]
It is easy to verify that
$F$ is $1$-Lipschitz and convex.

\paragraph{First-order oracle} Given $\bA$, $\bv_i$ and the function $F$ they define, we will use $\bg$ below as the subgradient part of the first-order oracle of $F$: on any query point $\bx\in \mB^d$, if the maximum is achieved at either
$
\frac{L\cdot \langle \bA_j,\bx\rangle-1}{\sqrt{d}L}
$
or $
\frac{-L\cdot \langle \bA_j,\bx\rangle-1}{\sqrt{d}L},
$
for some row $\bA_j$ of $\bA$ ($j\in [d/2]$), $\bg$ returns either $\bA_j/\sqrt{d}$ or $-\bA_j/\sqrt{d}$ accordingly with the smallest such $j\in [d/2]$;
otherwise, $\bg$ returns $\bv_i/(\sqrt{d}L)$ with the smallest $i\in [N]$ that achieves the maximum. 
In the rest of the section, when $F$ is the input function, the algorithm has access to $F$ and $\bg$ for subgradients of $F$ as the first-order oracle.

Our main result in this section is the following lemma, which shows that any query-efficient randomized convex optimization algorithm would imply an efficient deterministic communication protocol for the correlated orthogonal vector game. 
\begin{lemma}[Reduction]
\label{lem:reduction}
Let  $\eps = 1/{(d^2L)}$. 
If there is a randomized convex optimization  algorithm that uses  $S=d^{2-\delta}$ bits of memory, makes  $T = d^{1+\delta/6 - o(1)}$ queries and finds an $\eps$-optimal point~with probability at least $2/3$, 
then there is a deterministic $(k, n)$-communication protocol that succeeds with probability at least $1/2$ for the $(s, \xi)$-correlated orthogonal vector game, with parameters
$$
 k = \frac{S}{d}=d^{1-\delta - o(1)},\quad
 n=\frac{40T}{N}= d^{1-o(1)},\quad 
 s = d^{1-\delta/2}\log^2d\quad
 \text{and}\quad \xi = \frac{2}{L}.
$$
\end{lemma}

\def\ALG{\textsf{ALG}}

Assume such a randomized convex optimization algorithm exists. Since randomized algorithms are distributions of deterministic algorithms, there must be a deterministic $S$-bit-memory, $T$-query algorithm that outputs an $\epsilon$-optimal point of $F$ with probability at least $2/3$ (over the randomness of $\bA$ and $\bv_i$).
Let $\ALG$ denote such an algorithm.
For convenience we assume that the last point that $\ALG$ queries is the same point it outputs at the end; it is without loss of generality since such a requirement can only increase the number of queries by one.
We use it to obtain a deterministic communication protocol for the correlated orthogonal vector game in the rest of the section.

Given a function $F$, we first use the execution of $\ALG$ on $F$ to define the \emph{correlation time}:

\begin{definition}[Correlation time]\label{def:corrtime}
For each $i \in [N]$, let $t_i \in [T] \cup \{\infty\}$ be the first time that $\emph{\ALG}$ submits a query $\bx_{t_i} \in \mB^d$  such that
\begin{enumerate}
\item  $\bx_{t_i}$ has a large correlation  with $\bv_i$: $|\langle \bx_{t_i}, \bv_i\rangle| \geq  
 {\gamma}/{4}$; and
\item $\bx_{t_i}$ is almost orthogonal to $\bA$: $\|\bA\bx_{t_i}\|_{\infty} \leq  \xi$.
\end{enumerate}
If no such query exists during the execution of $\emph{\ALG}$, then we set $t_i = \infty$.
\end{definition}

First we show that, with high probability, the correlation time are in order $t_1 \leq t_2 \leq \cdots \leq t_{N}$:
\begin{lemma}
\label{lem:first-time}
Fix any  $\bA \in \{-1,1\}^{ ({d}/{2})\times d}$. With probability at least $1-d^{-\omega(1)}$ over the randomness of $\bv_1, \ldots, \bv_N$, we have
$t_1 \leq t_2 \leq \cdots \leq t_{N}$. 
\end{lemma}

\begin{proof}
Consider the process of first initializing $t_1=\cdots=t_N=\infty$ and then in each round $t\in [T]$, setting $t_i=t$ if $t_i=\infty$ at the end of round $t-1$ but  the point $\bx_{t}$ queried by $\ALG$ in round $t$ satisfies both conditions 
  $|\langle \bx_t,\bv_i\rangle |\ge \gamma /4$ and $\|\bA \bx_t\|_\infty\le \xi$.
Note that if the event 
  of the lemma is violated,
  then there must exist $t\in [T]$ and $i\in [N-1]$ such that the following two events hold:
\begin{enumerate}
    \item Event $E_1$: At the end of round $t-1$, we have $t_i=t_{i+1}=\infty$; and
    
    \item Event $E_2$: At the end of round $t$, we have $t_i=\infty$ but $t_{i+1}=t$.
\end{enumerate}
Fixing $t\in [N]$ and $i\in [N-1]$,
  we show in the rest of the proof that
  the probability of $E_1\land E_2$
  is at most $d^{-\omega(1)}$. The lemma  follows from a union bound on $t$ and $i$.

We start the following simple claim.

\begin{claim}\label{claim:hehe}
Assuming $E_1$, $\emph{\ALG}$ never receives
  $\bv_{i+1}/(\sqrt{d}L)$ as the subgradient
  before round $t$.
\end{claim}
\begin{proof}
Given that $t_i=t_{i+1}=\infty$ at the end of round $t-1$, every point $\bx$ queried by $\ALG$ before round $t$ satisfies one of the following conditions: either
\begin{enumerate}
    \item $\|\bA \bx\|_\infty> \xi$, 
    in which case the maximum in $F(\bx)$
    must be achieved by $L\|\bA\bx\|_\infty-1$; or 
    \item $|\langle \bx,\bv_i\rangle|<\gamma/4$ and $|\langle\bx,\bv_{i+1}\rangle|<\gamma /4$, in which case we have
$$
\langle \bx,\bv_{i+1}\rangle -(i+1)\gamma
< (\gamma/4)-(i+1)\gamma <-(\gamma/4) -i\gamma<\langle \bx,\bv_{i }\rangle -(i+1)\gamma.
$$
\end{enumerate}
So in both cases, $\bv_{i+1}/(\sqrt{d}L)$
  cannot be a subgradient of $F$ returned at $\bx$.
\end{proof}

We finish the proof by showing that 
   the probability of $E_2$ conditioning on
   $E_1$ is at most $d^{-\omega(1)}$.
To this end, consider running $\ALG$ on $F$
  that is defined using the fixed $\bA$ and random $\bv_1,\ldots,\bv_N$,
  and $E_1$ holds at the end of round $t-1$. 
Let $H$ denote the query-answer history of
  the $t-1$ rounds so far.
Let $\bx_1,\ldots,\bx_{t-1}$ be the $t-1$ points queried in $H$, and
  let $\bx_t$ be the point to be queried by $\ALG$ next in round $t$ given $H$.
By Claim \ref{claim:hehe}, $\bv_{i+1}/\sqrt{d}L$ is never returned as a subgradient in $H$.

Before proceeding to round $t$, we 
  reveal all vectors $\bv_1,\ldots,\bv_N$ except
  $\bv_{i+1}$.
Let $\bv_{-i}$ denote the tuple that contains
  these $N-1$ vectors revealed. 
We claim that given $\bv_{-i}$, the following 
  set of vectors captures exactly candidates for $\bv_{i+1}$ to (i) be consistent with $H$ and (ii) satisfy $E_1$:
\begin{align}
V := \left\{\bv \in  \frac{1} {\sqrt{d}}\mH_d: |\langle \bv, \bx_{\tau}\rangle| \leq \frac{\gamma}{4}\ \text{for all}\ \tau \in [t-1]\ \text{with}\ 
\|\bA \bx_\tau\|_\infty\le \xi\right\}. \label{eq:reduction3}
\end{align}
To see this is the case, it is trivial that
  $\bv_{i+1}$ needs to be in $V$ for both (i) and (ii) to hold.
On the other hand, $\bv_{i+1}\in V$ is sufficient for (i) and (ii) since no subgradient in $H$ is relevant to $\bv_{i+1}$.

It then suffices to show that 
  the probability of $\bv\sim V$ satisfying 
  $|\langle \bv,\bx_{t}\rangle|\ge \gamma /4$ 
  is at most $d^{-\omega(1)}$.
First we have 
by Khintchine inequality (Lemma \ref{lem:khintchine}) and using $\delta <1$ that
\begin{align*}
\Pr_{\bv\sim \frac{1}{\sqrt{d}}\mH_d} \left[\bv \in V \right] 
\geq 1 - d^{-\omega(1)}.
\end{align*}
On the other hand, we have by Chernoff bound that 
$$
\Pr_{\bv\sim \frac{1}{\sqrt{d}}\mH_d } \left[
|\langle \bv,\bx_{t}\rangle| \ge \gamma /4
\right]\le d^{-\omega(1)}.
$$ 
It follows that the probability of $\bv\sim V$ satisfying  
  $|\langle \bv,\bx_{t}\rangle|\ge  \gamma /4$ 
  is at most $d^{-\omega(1)}$.
\end{proof}

Next, we observe the optimal value of $F$ is small with high probability:
\begin{lemma}
\label{lem:obj-value}
Fix $\bA \in \{-1,1\}^{(d/2)\times d}$. With probability at least $1-d^{-\omega(1)}$ over  $\bv_1, \ldots, \bv_N$ we have 
\[
\argmin_{\bx \in \mathbb{B}^d}F(\bx ) \leq -\frac{1}{\sqrt{d}L}\cdot \frac{1}{\sqrt{N}\log^2 d}.
\]
\end{lemma}
\begin{proof}
Let $\bU_\bA$ be an orthonormal matrix that spans the row space of $\bA$.
Let $\hat{\bv}_i = (\bI - \bU_\bA\bU_\bA^\top)\bv_i$ be projection $\bv_i$ onto the orthogonal space of $\bU_\bA$. 
It is clear that $\hat{\bv}_i \perp \row(\bA)$ for any $i \in [N]$.

Consider
\[
\bx = -\frac{1}{\sqrt{N}\log d}\cdot \sum_{i=1}^{N} \hat{\bv}_{i}.
\]
It suffices to prove that, with probability at least $1 -d^{-\omega(1)}$, one has
\begin{itemize}
\item $F(\bx) \leq -\frac{1}{\sqrt{d}L} \cdot \frac{1}{\sqrt{N}\log^2 d}$, and
\item $\|\bx\|_2 \leq 1$.
\end{itemize}
For the first claim, it is clear that 
$\|\bA\bx\|_\infty = 0$ since $\hat{\bv}_i$ ($i \in [N]$) is orthogonal to the row space of $\bA$.
Furthermore, for any $i \in [N]$, one has 
\begin{align}
\langle \bv_i, \bx\rangle = &~ \bv_i (\bI-\bU_\bA \bU_\bA^\top)\bx = \langle \hat{\bv}_i, \bx\rangle \notag\\
=&~ \left\langle \hat{\bv}_{i}, - \frac{1}{\sqrt{N}\log (d)}\sum_{j=1}^{N} \hat{\bv}_{i}\right\rangle \notag\\
\leq &~  -\frac{1}{\sqrt{N}\log d} \|\hat{\bv}_{i}\|_2^2 + \frac{1}{\sqrt{N}\log d}  \left| \sum_{j\neq i} \langle \hat{\bv}_{i}, \hat{\bv}_{j} \rangle\right|. \label{eq:obj1}
\end{align}
The first step holds since $(\bI-\bU_\bA \bU_\bA^\top)\bx = \bx$, the second and third step follow from the definition of $\hat{\bv}_i$ and $\bx$.

We bound the two terms separately. For the first term, since $\hat{\bv}_i$ is the projection of $\bv_i$ onto the orthogonal space of $\bU_\bA$ (which has rank at least $d/2$), by Lemma \ref{lem:projection-prob}, with probability at least $1- d^{-\omega(1)}$, one has
\begin{align}
\|\hat{\bv}_{i}\|_2^2 \geq \frac{1}{2} - \frac{\log d}{\sqrt{d}}.\label{eq:obj2}
\end{align}
For the second term,  with probability at least $1- d^{-\omega(1)}$, one has 
\begin{align}
 \left\langle \hat{\bv}_{i}, \sum_{j\neq i} \hat{\bv}_{j} \right\rangle = &~ \bv_i^\top (\bI-\bU_\bA\bU_\bA^\top) (\sum_{j\neq i} \bv_{j}) \notag \\
 \leq &~ \frac{\log d}{\sqrt{d}} \cdot \|(\bI-\bU_\bA \bU_\bA^\top) (\sum_{j\neq i} \bv_{j})\|_2 \notag \\
  \leq &~ \frac{\log d}{\sqrt{d}} \cdot \|\sum_{j\neq i} \bv_{j}\|_2 \notag \\
 \leq &~ \frac{\log d}{\sqrt{d}}  \cdot \sqrt{N}\log d \leq \frac{1}{4}.\label{eq:obj3}
\end{align}
Here the second step holds due to Khintchine inequality (Lemma \ref{lem:khintchine}), the fourth step holds due to Chernoff bound, and the last step holds due to the choice of parameters.

Combining Eq.~\eqref{eq:obj1}\eqref{eq:obj2}\eqref{eq:obj3}, we conclude that 
\[
\langle \bv_i, \bx\rangle \leq - \frac{1}{\sqrt{N}\log^2 (d)} \quad \forall i \in [N]
\]
and therefore, complete the proof of the first claim.

For the second claim, with probability at least $1-d^{-\omega(1)}$
\begin{align*}
\|\sum_{i=1}^{N}\hat{\bv}_{i}\|_2 = \|(\bI -\bU_\bA \bU_\bA^\top) \sum_{i=1}^{N}\bv_i\|_2 \leq \|\sum_{i=1}^{N}\bv_i\|_2 \leq \sqrt{N}\log(d)
\end{align*}
We finish the proof here.
\end{proof}

It follows from  Lemma \ref{lem:obj-value}
  that $t_N\le T$ with probablity at least $(2/3) - d^{-\omega(1)}$:
\begin{lemma}
\label{lem:success}
With probability at least $(2/3) - d^{-\omega(1)}$, we have $t_{N} \leq T$. 
\end{lemma}
\begin{proof}
We claim that $t_N\le T$ whenever (1) $\ALG$ finds 
  an $\eps$-optimal point of $F$ and 
  (2) the event of Lemma \ref{lem:obj-value} holds.
The lemma follows from
  Lemma \ref{lem:obj-value} and the assumption that $\ALG$ succeeds with
  probability at least $2/3$.
To prove the claim, we note that by (2), 
\[
\argmin_{\bx \in \mathbb{B}^d}F(\bx ) \leq -\frac{1}{\sqrt{d}L}\cdot \frac{1}{\sqrt{N}\log^2d}.
\]
On the other hand, if $t_N=\infty$, then 
  every point $\bx_t$ queried by $\ALG$ satisfies
\[
F(\bx_t) \geq \frac{1}{\sqrt{d}L}\cdot \max\Big\{L\|\bA\bx_{t}\|_\infty-1, \langle \bv_N, \bx_t\rangle - N\gamma \Big\} \geq \frac{1}{\sqrt{d}L}\cdot \left(-N - \frac{1}{4}\right) \gamma \geq \argmin_{\bx \in \mathbb{B}^d}F(\bx ) +\eps.
\]
Here the second step follows from $t_N=\infty$ and the last step follows from our choice of parameters. 
Given that we assume $\ALG$ outputs $\bx_N$, this contradicts with
  the assumption that $\ALG$ succeeds in finding an $\eps$-optimal point of $F$.
This finishes the proof of the lemma.
\end{proof}

By a simple averaging argument, there exists 
  an $i\in [N-1]$ with the following property:
\begin{lemma}
\label{lem:gap}
There exists an $i\in [N-1]$ such that both $t_i\le t_{i+1}\le T$ and
\[
t_{i + 1} - t_{i} \leq \frac{20T}{N} =\frac{n}{2} 
\]
hold with probability at least $3/5$.
\end{lemma}
\begin{proof}
We always condition on the high probability event of Lemma \ref{lem:first-time}. 
Let $\mE$ denote the event that $t_N \leq T$ and by Lemma \ref{lem:success}, we have $\Pr[\mE]\geq \frac{2}{3}-d^{-\omega(1)}$.
We prove by contradiction and suppose there is no such index $i$, then we have $\Pr[t_{i + 1} - t_{i} \geq \frac{20T}{N}] > \frac{3}{5}$ for all $i\in [N-1]$, and therefore, $\Pr[t_{i + 1} - t_{i} \geq \frac{20T}{N} | \mE] > \frac{1}{15}$ . This means $\E[t_{i + 1} - t_{i} | \mE] > \frac{1}{15} \cdot \frac{20T}{N} \geq \frac{4T}{3N}$, and 
\[
\E\left[t_{N} - t_{1} | \mE\right] = \E\left[\sum_{i=1}^{N-1} t_{i + 1} - t_{i} | \mE\right] > (N-1) \cdot \frac{4T}{3N} > T
\]
This contradicts the fact that $t_N \leq T$ when $\mE$ happens.
\end{proof}

We remark that all lemmas proved so far stand regardless of the memory constraints. \medskip

\noindent\textbf{Reduction.} Now we are ready to use $\ALG$ to design a randomized $(k,n)$-communication protocol for the  $(s,\xi)$-correlated orthogonal vector game that succeeds with probability at least $1/2$. Given that a randomized protocol is a distribution of deterministic protocols, Lemma \ref{lem:reduction} follows.

\begin{proof}[Proof of Lemma \ref{lem:reduction}]
Let $i\in [N-1]$ be an integer that 
  satisfies the condition of Lemma \ref{lem:gap}. 
The communication protocol is described in Figure~\ref{fig:reduction}.
Recall that in the game Alice receives a matrix $\bA \sim \{-1, 1\}^{ ({d}/{2}) \times d}$ and Bob receives a vector $\bv \sim \smash{ \frac{1}{\sqrt{d}}\mH_{d}}$. In the protocol, Alice and Bob  sample independently $\bv_{1}, \ldots, \bv_N \sim \smash{\frac{1}{\sqrt{d}}\mH_d}$ using public randomness. Let 
\[
V = (\bv_1, \ldots, \bv_i, \bv_{i+1}, \bv_{i+2}, \ldots, \bv_{N}) \quad \text{and} \quad  V' = (\bv_1, \ldots, \bv_i, \bv, \bv_{i+2}, \ldots, \bv_{N}),
\]
that is, $V'$ replaces $\bv_{i+1}$ with $\bv$.
Let $F_{\bA,V}$ (or $F_{\bA,V'}$) denote the function defined using $\bA$ and those vectors in $V$ (or $V'$, respectively). Roughly speaking, Alice would first optimize the function $F_{\bA, V}$ to time $t_i$, then send its memory state $\bM$ to Bob. It continues to optimize $F_{\bA, V'}$ for $n$ iterations and sends all subgradients to Bob. 
Bob would simulate $\ALG$ using $\bM$ and these subgradients, output any query $\bx$ that is orthogonal to $\bA$ and has large correlation with $\bv$.

\begin{figure}[!t]
\begin{tcolorbox}[standard jigsaw, opacityback=0,title=Communication protocol for the correlated orthogonal vector game] 
\begin{flushleft}\begin{itemize}
\item Alice and Bob use public randomness to sample $\smash{\bv_1, \ldots, \bv_{N} \sim \frac{1}{\sqrt{d}}\mH_d}$. 
\item Round 1: Alice runs $\ALG$ on the function $F_{\bA, V}$, from which Alice obtains $t_i$ and sends the memory state $\bM \in \{0,1\}^{kd}$ of $\ALG$ after $t_i-1$ rounds to Bob.
\item  Round 2: Bob sends $\bv$ to Alice.
\item Round 3: Alice runs $\ALG$ on the function $F_{\bA,V'}$.
For each $j\in [n]$, Alice sets $\ba_j$ to be
  $\bA_i$ if the subgradient returned by $F_{\bA,V'}$ in 
  round $t_i+j-1$ is either $\bA_i/\sqrt{d}$ or $-\bA_i/\sqrt{d}$,
  and set $\ba_j=\nil$ otherwise.
Alice sends $\ba_1,\ldots,\ba_n\in \row(\bA)\cup \{\nil\}$ to Bob.

\item Output: Bob runs $\ALG$ for $n$ rounds, starting with memory state $\bM$,  
 the first message from Alice.
 For each round $j=1,\ldots,n$, let $\bx_j$ be the query point of $\ALG$:
\begin{enumerate}
    \item If $\ba_j\ne \nil$, Bob sets $$\left( \frac{L|\langle \ba_j,\bx_j\rangle|-1}{ \sqrt{d}L},\hspace{0.06cm} \pm \frac{\ba_j}{\sqrt{d}}\right)$$
    be the pair returned by the oracle, where the sign of the subgradient can be determined by the sign of $\langle \ba_j,\bx_j\rangle$, and continue the execution of $\ALG$;
    \item  If $\ba_j=\nil$, Bob finds the vector $\bv_{\ell}\in V'$ with the smallest index that maximizes $\langle \bv_\ell,\bx_j\rangle -\ell\gamma$, sets
    $$\left(\frac{\langle \bv_\ell,\bx_j\rangle -\ell\gamma}{\sqrt{d}L},\hspace{0.06cm}\frac{\bv_\ell}{\sqrt{d}L}\right)$$
    be the pair returned by the oracle and continue the execution of $\ALG$.
\end{enumerate}
Bob outputs any query $\bx_j$ ($j \in [n]$) that satisfies $|\langle \bv,\bx_j\rangle| \geq \gamma /4$ and $F(\bx_j) \leq \frac{1}{\sqrt{d}L}$.
If no such query exists, Bob fails and outputs an arbitrary vector, say $\mathbf{0}$.
\end{itemize}
\end{flushleft}
\end{tcolorbox}
\caption{Reduction from convex optimization to the correlated orthogonal vector game}\label{fig:reduction}
\end{figure}

It is clear from the description of the protocol that it is a $(k,n)$-communication protocol. 
We finish the proof by showing that it succeeds with probability at least $1/2$.
We first condition on the high probability event of Lemma \ref{lem:first-time}, for $V$ and $V'$. Note it happens with probability at least $1-d^{-\omega(1)}$. 
The convex optimization algorithm has the same transcripts for $F_{\bA,V}$ and $F_{\bA, V'}$, during time $t \in [t_{i}-1]$. 
This is because it only receives subgradients from $\{\bv_1, \ldots, \bv_{i}\} \cup \row(\bA)$ before time $t_{i}$ by Lemma \ref{lem:first-time} and Claim \ref{claim:hehe}.
Hence, we can assume the memory state $\bM$ comes from optimizing the function $F_{\bA, V'}$.

Next, we condition on the event of Lemma \ref{lem:gap}, which asserts $t_{i+1} - t_i \leq n/2$ and it happens with probability at least $3/5$.
It is easy to see that Bob, who knows $\bM$, $V'$ and $\ba_{1}, \ldots, \ba_{n}$, can simulate $\ALG$ for function $F_{\bA, V'}$ up to time $t_{i+1} \leq t_i + n$. 
Therefore, it must submit a query $\bx$ that satisfies $\|\bA\bx\|_\infty \leq \frac{2}{L} = \xi$ and $|\langle \bx, \bv\rangle| \geq \gamma/4 \geq \sqrt{s/d}$. We finish the proof of reduction here.
\end{proof}


\section{Lower bound for the correlated orthogonal vector game}
\label{sec:lower}

In this section, We prove the following lower bound for correlated orthogonal vector game.
Theorem \ref{theo:main} follows directly by combining the lower bound with Lemma \ref{lem:reduction}.

\begin{theorem}[Lower bound of correlated orthogonal vector game]
Let $d$ be a sufficiently large integer. Let $\delta\in (0,1)$, $k = d^{1-\delta - o(1)}$, $s = d^{1-\delta/2}\log^2(d)$, $\xi = 2\exp(-\log^5(d))$ Then any deterministic $(k,n)$-communication protocol that solves the $(s,\xi)$-correlated orthogonal vector game with probability at least $1/2$ requires 
$
n \geq d\cdot \exp(-\log d /\log\log d).
$
\end{theorem}

We prove by contradiction and assume $n \leq d\cdot \exp(-\log d /\log\log d)$ from now on.

\paragraph{Communication protocol.} A deterministic $(k,n)$-communication protocol works as follows:
\begin{flushleft}\begin{itemize}
\item (Round 1) Alice receives a matrix $\bA \in \{-1,1\}^{({d}/{2})\times d}$, and let $\bM_{\bA} \in \{0,1\}^{kd}$ be the message sent in the first round.
\item (Round 3) Recall Alice obtains both $\bA$ and $\bv$ after the second round. Let $$\br_{i, \bA, \bv} \in \{-1,1\}^d \cup \{\nil\}$$ be the $i$-th row sent to Bob ($i\in [n]$) and let $\bR_{\bA,\bv} = (\br_{1, \bA,\bv}, \ldots, \br_{n, \bA,\bv})$ be the collection of rows sent by Alice.
\item (Output) At the end, Bob outputs a vector $\bx_{\bM, \bv, \bR}$ based on the message $\bM$, the vector $\bv$ and the collection of rows $\bR$. We write $\bx_{\bA, \bv} = \bx_{\bM_{\bA},\bv,\bR_{\bA, \bv}}$ to denote Bob's output under the input pair $(\bA, \bv)$.
\end{itemize}\end{flushleft}

Note the output vector $\bx_{\bM, \bv, \bR}$ is well-defined for any message $\bM\in \{0,1\}^d$, vector $\bv\in \frac{1}{\sqrt{d}}\mH_d$ and any collection of rows $\bR \in (\{-1,1\}^d\cup \{\nil\})^{n}$.

The following lemma shows that, without loss of generality, one can assume that the output $\bx_{\bM, \bv, \bR}$ has unit norm.
\begin{lemma}[Unit norm]
If there is a $(k, n)$ communication protocol for $(s, \xi)$-correlated orthogonal vector game, then there is a $(k, n)$ communication protocol that always output unit vectors and solves $(s, \xi')$-correlated orthogonal vector game, where $\xi' = \sqrt{d}\xi$
\end{lemma}
\begin{proof}
For any message $\bM\in \{0,1\}^{kd}$, vector $\bv\in \frac{1}{\sqrt{d}}\mH_d$ and collection of rows $\bR \in (\{-1,1\}^d\cup \{\nil\})^{n}$, consider the output vector $\bx_{\bM, \bv, \bR}$:
\begin{itemize}
\item If $\|\bx_{\bM, \bv, \bR}\|_2 < \sqrt{s/d}$, one can change it to an arbitrary unit vector and it does not reduce the success probability of the protocol.
\item If $\sqrt{s/d} \leq \|\bx_{\bM, \bv, \bR}\|_2 \leq 1$, then we can scale the output to $\frac{\bx_{\bM, \bv, \bR}}{\|\bx_{\bM, \bv, \bR}\|_2}$ and it has unit norm.
The correlation with any vector can only increase, and by a factor of at most $\sqrt{d/s} \leq \sqrt{d}$, hence by relaxing the orthogonal condition by a factor of $\sqrt{d}$, the success probability does not decrease.
\end{itemize}
We complete the proof here.
\end{proof}

For any matrix $\bB \in \{-1,1\}^{m\times d}$ and any collection of rows $\bR \in (\{-1,1\}^d \cup \{\nil\})^{n}$, we write $\bR \subseteq \row(\bB) \cup \{\nil\}$ if each row of $\bR$ either is $\nil$ or  belongs to $\row(\bB)$.

\begin{definition}[Table]
\label{def:table}
Let $m \in [d/2]$. Given any message $\bM \in \{0, 1\}^{kd}$, any matrix $\bB \in \{-1,1\}^{m\times d}$ and any set of vectors $V \subseteq \frac{1}{\sqrt{d}} \mH_{d}$ 
\[
\T^{\bM}(\bB, V) :=  \left\{\bx_{\bM, \bv, \bR} :  \bv \in V, \bR\subseteq \row(\bB)\cup \{\nil\}\right\} \subseteq \bS^{d}
\]
That is to say, $\T^{\bM}(\bB,V)$ contains all possible outputs by Bob, when the message is $\bM$, the vector $\bv$ is from $V$ and the n rows are taken from $\bB$.
\end{definition}
Note the definition of table is flexible that we do not need $\bB$ to have exact $d/2$ rows. Given a table $\T^{\bM}(\bB, V)$, we mostly care about entries that are (almost) orthogonal to $\bB$.
Recall $\xi' = \sqrt{d}\xi$.
\begin{definition}[Orthogonal entry]
Given a table $\T^{\bM}(\bB, V)$, let $\mO^{\bM}(\bB, V)$ contain all entries in $\T^{\bM}(\bB, V)$ that are (almost) orthogonal to $\bB$, i.e.,
\begin{align*}
\mO^{\bM}(\bB, V) := \left\{\bx: \bx \in \T^{\bM}(\bB, V), \|\bB\bx\|_\infty \leq \xi' \right\}.
\end{align*}
\end{definition}


\subsection{The existence of a succinct table}

The goal of this subsection is to select a small subset $V^{*} \subseteq \frac{1}{\sqrt{d}}\mH_d$, of size at most $2^s$, that is representative.

Let $\mA_{\suc} \subseteq \{-1,1\}^{\frac{d}{2}\times d}$ contain all matrices $\bA$ such that the protocol succeeds with probability at least $1/4$ (over the randomness of $\bv$) when Alice receives $\bA$ as input. We have $|\mA_{\suc}| \geq \frac{1}{4}\cdot 2^{d^2/2}$ because the protocol succeeds with probability at least $1/2$ over a random pair of $\bA$ and $\bv$. 
First, we prove the outputs $\{\bx_{\bA, \bv}\}_{\bv \in \frac{1}{\sqrt{d}}\mH_d}$ are spread out for $\bA \in \mA_{\suc}$.

\begin{lemma}
\label{lem:spread}
Let $c_1 > 0$ be some sufficiently small constant and $K = \exp(c_1 s)$.
For any fixed matrix $A \in \mA_{\suc}$ and fixed set of vectors $\by_1, \ldots, \by_{K} \in \bS^{d}$, one has
\begin{align}
\Pr_{\bv\sim \frac{1}{\sqrt{d}}\mH_d}\left[\|\bA \bx_{\bA, \bv}\|_\infty \leq \xi' \wedge \|\bx_{\bA, \bv} - \by_i\|_2 > \frac{1}{d^8}\, \forall i\in [K]\right] \geq 1/8.
\label{eq:spread}
\end{align}
\end{lemma}
\begin{proof}
Let $V_{\close} \subseteq \frac{1}{\sqrt{d}}\mH_d$ be the set of vectors $\bv$ that has large correlations with some $\by_i$, i.e.,
\begin{align*}
V_{\close} = \bigcup_{i \in [K]} \left\{\bv: |\langle \bv, \by_i\rangle| \geq \sqrt{s/2d}, \bv\in \frac{1}{\sqrt{d}}\mH_d\right\}.
\end{align*}
We first prove the size of $V_\close$ is not large. 
For any fixed $\by_i$ ($i \in [K]$), by Lemma \ref{lem:projection-prob}, we have that
\begin{align*}
\Pr_{\bv\sim \frac{1}{\sqrt{d}}\mH_d}\left[|\langle \bv, \by_i\rangle| \geq \sqrt{s/2d}\right] \leq \exp(-c_2 s)
\end{align*}
for some constant $ c_2 > c_1  > 0$. Taking a union bound over $i \in [K]$, we have
\begin{align*}
\Pr_{\bv\sim \frac{1}{\sqrt{d}}\mH_d}\left[\exists i \in [K]: |\langle \bv, \by_i\rangle| \geq \sqrt{s/2d} \right] \leq K \cdot \exp(-c_2 s) \leq \exp(-(c_2-c_1)s) \leq \frac{1}{8}
\end{align*}
Hence, we have $|V_{\close}| \leq \frac{1}{8} \cdot 2^d$. 

Let $V' \subseteq \frac{1}{\sqrt{d}}\mH_d$ be the set of vectors such that Eq.~\eqref{eq:spread} holds. It suffices to prove that the protocol succeeds on $(\bA, \bv)$ only if $\bv \in V_{\close} \cup V'$, as this would imply $|V'| \geq \frac{1}{8} \cdot 2^d$.
For any vector $\bv \in \frac{1}{\sqrt{d}} \mH_d\backslash (V_{\close} \cup V')$, if $\|\bA \bx_{\bA,\bv}\|_\infty \leq \xi'$, then $\bx_{\bA, \bv}$ must be close to some $\by_i$ (since $\bv \notin V'$). Then we have 
\[
|\langle \bx_{\bA, \bv}, \bv\rangle| \leq |\langle \by_i, \bv\rangle| + \frac{1}{d^8} \leq \sqrt{\frac{s}{2d}} + \frac{1}{d^8} < \sqrt{\frac{s}{d}}
\]
where the first step follows from $\|\bx_{\bA, \bv} - \by_i\|_2 \leq \frac{1}{d^8}$, the second step follows from $\bv \notin V_{\close}$. We finish the proof here.
\end{proof}

\begin{definition}[$\alpha$-cover]
Let $\alpha > 0$ and $X \subseteq \R^{d}$ be a set of points. 
Let $\mN(X,\alpha) \subseteq X$ be the $\alpha$-covering of set $X$, which is defined as the {\em largest} set $X'\subseteq X$ such that for any $\bx_1, \bx_2 \in X'$, $\|\bx_1 - \bx_2\|_2 \geq \alpha$.
\end{definition}

Now, we prove that there exists a ``small'' set of vectors $V^{*} \subseteq \frac{1}{\sqrt{d}}H_d$ , such that the $(1/d^8)$-cover of $\mO^{\bM_\bA}(\bA, V^{*})$ is large , for most $\bA \in \mA_{\suc}$. 
We derive the existence using the probabilistic method. 
\begin{lemma}
\label{lem:probabilistic-v}
Let $c_1 > 0$ be a sufficiently small constant and $K = \exp(c_1 s)$. 
For any fixed matrix $\bA \in \mA_{\suc}$, with probability least $1 - \exp(-s/2)$ over the random draw of $V = \{\bv_{1}, \ldots \bv_{16sK}\}$,  
\[
\left|\mN\left(\mO^{\bM_\bA}(\bA, V), 1/d^8\right)\right|\geq K.
\]
\end{lemma}
\begin{proof}
We partition the set $V$ into $K$ groups, where the $i$-th  group 
\[
V_{i} = \{\bv_{16s(i-1)+1}, \ldots, \bv_{16si}\} \quad \forall i \in [K]
\]
We would construct a large set $\mN$ from $\mO^{\bM_{\bA}}(\bA, V)$ such that the pairwise distance is large.
Initially, $\mN_0 = \emptyset$.  
At the $i$-th step, suppose $\mN_{i-1} = \{\bx_1, \ldots, \bx_{i-1}\} \subseteq \mO^{\bM_{\bA}}(\bA, V_1 \cup \cdots \cup V_{i-1})$ be the set constructed thus far, we would add an entry from $\mO^{\bM_\bA}(\bA, V_{i})$ that is apart from existing vectors in $\mN_{i-1}$.

Note for each $j \in [16s]$, consider the entry $\bx_{\bA, \bv_{16s(i-1) + j}} \in \T^{\bM_{\bA}}(\bA, V_i)$, by Lemma \ref{lem:spread}, we have
\begin{align*}
\Pr_{\bv_{16s(i-1)+j}\sim \frac{1}{\sqrt{d}} \mH_d}\left[\|\bA \bx_{\bA, \bv_{16s(i-1)+j}}\|_\infty \leq \xi' \wedge \|\bx_{\bA, \bv_{16s(i-1)+j}} - \bx_\tau\|_2 > \frac{1}{d^8}\, \forall \tau \in [i-1] \right] \geq 1/8.
\end{align*}
Since $\{\bv_{16s(i-1) + j}\}_{j \in [16s]}$ are sampled independently  from $\frac{1}{\sqrt{d}} \mH_d$, with probability at least $1 -\exp(-s)$, there exists $j\in [16s]$ such that
\begin{itemize}
\item $\|\bx_{\bA, \bv_{16s(i-1) + j}} - \bx_\tau\|_2 \geq \frac{1}{d^8}$ for any $\tau\in [i-1]$, and
\item $\|\bA\bx_{\bA, \bv_{16s(i-1) + j}}\|_\infty \leq \xi'$.
\end{itemize}
Hence we can add $\mN_{i} = \mN_{i-1} \cup \{\bx_{\bA, 16s(i-1), j}\}$. Taking a union bound over $i\in [K]$, we have with probability at least $1 - \exp(-s/2)$, one has $|\mN(\mO^{\bM_{\bA}}(\bA, V), 1/d^8)| \geq K$.
\end{proof}

By Lemma \ref{lem:probabilistic-v}, we conclude that
\begin{lemma}
\label{lem:exist-v}
There exists a set $V^{*} \subseteq \frac{1}{\sqrt{d}} \mH_d$ with size at most $16sK \leq 2^{s}$,  such that the set
\begin{align*}
\nA := \left\{ \bA \in \mA_{\suc}: |\mN(\mO^{\bM_{\bA}}(\bA, V^{*}), 1/d^8)| \geq 2^{s/\log (d)}\right\}. 
\end{align*}
has size at least $\frac{1}{8} \cdot 2^{d^2/2}$.
\end{lemma}

We would fix the set of $V^{*}$ from now on, and we only consider matrix $\bA \in \nA$. In the rest of the proof, we would omit $V^{*}$ when there is no confusion. That is to say, we write the table $\T^{\bM}(\bA) := \T^{\bM}(\bA, V^{*})$ and its orthogonal entries $\mO^{\bM}(\bA) := \mO^{\bM}(\bA, V^{*})$.

\subsection{Encoding}

Up to this point, we have proved that there exists a large set of matrices $\nA \subseteq \{-1,1\}^{(d/2)\times d}$, such that, for any $\bA \in \nA$, the table $\T^{\bM_{\bA}}(\bA)$ contains many different orthogonal entries: The $(1/d^8)$-cover of $\mO^{\bM_\bA}(\bA)$ is of size at least $2^{s/\log (d)}$. 
We establish the contradiction by proving this is simply impossible!
Our strategy is to find an encoding strategy such that (1) it encodes every matrix in $\nA$, but at the same time (2) the number of matrices encoded are at most $2^{d^2/2}/d$.

We first introduce the notion of robustly linearly independent sequence.\footnote{Our definition of $\gamma$-robustly linearly independent sequence is slightly different from \cite{marsden2022efficient}, roughly speaking, $\gamma$-RLI in our work implies $(\gamma^2/2)$-RLI of \cite{marsden2022efficient} and vice versa. }
\begin{definition}[$\gamma$-robustly linearly independent sequence]
\label{def:robust-linear-independent}
Let $\gamma \in (0,1), L \in [d]$. A sequence of unit vectors $\by_1, \ldots, \by_L$ is $\gamma$-robustly linearly independent, if for any $j \in [L-1]$,
\[
\|\by_j - \proj_{\mathsf{span}(\by_1, \ldots, \by_{j-1})}(\by_j)\|_2 \geq \gamma
\]
That is to say, $\by_j$ has a non-trivial component that is orthogonal to the linear span of $\by_1, \ldots, \by_{j-1}$.
\end{definition}

The definition of $\gamma$-RLI sequence allows for arbitrary length, but in the rest of section, we would consider sequence of length
\[
L := \frac{s}{4\log^8 d}
\]

\paragraph{Parameters.} We introduce a few parameters. Let 
\[
\Delta := \frac{d}{n} \geq \exp(\log d /\log \log d).
\]
Let $H$ be the smallest integer such that
\[
\frac{s}{\log^2(d)} \cdot \left(\frac{\Delta}{\log^5(d)}\right)^{H} \geq \frac{d}{10},
\]
we know that $H \leq \log\log d$ by our choice of parameters.

The encoding strategy proceeds in $H$ levels, and the step size 
\[
s_h =  \frac{s}{\log^2(d)} \cdot \left(\frac{\Delta}{\log^5(d)}\right)^{h}, 
\quad \forall h\in [0: H-1]  \quad \text{and} \quad s_H = \frac{d}{10}.
\]
Denote the partial sum as $s_{\leq h} = \sum_{i=1}^{h}s_i$ and $s_{\leq 0} = 0$ (note they exclude $s_0$).

We consider doubly exponentially decreasing radius of cover $\{\alpha_h\}_{h \in [H]}$,
\[
\alpha_0 = \frac{1}{d^8}, \quad \alpha_h = \alpha_{h-1}^{8} = d^{-8^{h+1}} \, \forall h \in [0:H]
\]
Note we still guarantee $\alpha_H \geq \exp(-8\log^4(d)) \gg \xi' = \sqrt{d}\exp(-\log^5(d))$.

\begin{definition}[Partition]
At the $h$-th level ($h\in [H]$), we consider partitions of rows $[d/2]$ that satisfy 
\[
P_{h, 1} \cup P_{h, 2} \cup P_{h, 3}= [d/2], \quad |P_{h, 1}| = d/2 - s_{\le h}, \, |P_{h, 2}| = s_{h},\,  |P_{h,3}| = s_{\le h-1}
\]
and let $\sP_h$ be the collection of all such partitions.

Given a partition $P_h \in \sP_h$ and three sub-matrices $\bA_{h, 1} \in \{-1,1\}^{(d/2-s_{\le h})\times d}$, $\bA_{h, 2} \in \{-1,1\}^{s_h\times d}$ and $\bA_{h, 3} \in \{-1,1\}^{s_{\le h-1}\times d}$,  
Let $[\bA_{h,1}, \bA_{h,2}, \bA_{h,3}]_{P_h} \in \{-1,1\}^{(d/2)\times d}$ be the matrix induced naturally by the partition. 
That is, the $j$-th row of $\bA_{h, \tau}$ is located at the $P_{h, \tau}(j)$-th row, where $P_{h, \tau}(j)$ is the $j$-th element in $P_{h, \tau}$ ($\tau =1,2,3, j \in [|P_{h, \tau}|]$).
\end{definition}

Finally, we set $\Gamma_h := 2^{s_h d - 2kd}$ to be the threshold for level $h$.



\paragraph{Encoding algorithm} The encoding strategy is formally stated at Algorithm \ref{algo:encode}. It divides into $H$ levels. In the $h$-th level, it enumerates all possible messages $\bM \in \{0,1\}^{kd}$, partitions $P_h \in \sP_h$ and sub-matrices $\bA_{h,1} \in \{-1,1\}^{(d/2-s_{\le h})\times d}$. 
It remains to determine $\bA_{h,2} \in \{-1,1\}^{s_h \times d}$ ($\bA_{h,3}$ could take any value in $\{-1,1\}^{s_{\le h-1} \times d}$).
To this end, it checks the orthogonal entries $\mO^{\bM}(\bA_{h, 1})$ and enumerates all $(\alpha_{h-1}/4)$-RLI sequence (of length $L$) in it. 
It includes matrices $\bA_{h, 2}$ that is orthogonal to one of the sequence and and takes a union over all of them (Line \ref{line:sequence-begin}-\ref{line:sequence-end}).
Finally, the encoding algorithm includes this set of matrices (i.e., $\calS_{\bM, P_h, \bA_{h,1}}$) only if its size is no more than $\Gamma_h =2^{s_h d -2kd}$.

\begin{algorithm}[!t]
\caption{Encoding}
\label{algo:encode}
\begin{algorithmic}[1]
\State $\mA_h \leftarrow \emptyset$ ($h \in [H]$) 
\For{$h=1,2,\ldots, H$}
\ForEach{message $\bM \in \{0, 1\}^{kd}$, partition $P_{h} \in \sP_h$, matrix $\bA_{h,1} \in \{-1,1\}^{(d/2-s_{\le h})\times d}$}
\State $\calS_{\bM, P_h, \bA_{h, 1}} \leftarrow \emptyset$
\ForEach{$(\alpha_{h-1}/4)$-RLI sequence $\bx_{h, 1}, \ldots, \bx_{h, L} \in \mO^{\bM}(\bA_{h, 1})$}   \label{line:sequence-begin}
\State $\mI_{h, \bx_{h, 1}, \ldots, \bx_{h, L}} \leftarrow \{\bA_{h, 2} \in \{-1,1\}^{s_h\times d}, \|\bA_{h,2} \bx_{h, i}\|_\infty \leq \xi' \, \forall i \in [L]\}$
\State $\calS_{\bM, P_h, \bA_{h, 1}} \leftarrow \calS_{\bM, P_h, \bA_{h, 1}} \cup \mI_{h, \bx_{h, 1}, \ldots, \bx_{h, L}}$
\EndFor\label{line:sequence-end}
\State $\mJ_{\bM, P_h, \bA_{h, 1}} \leftarrow \{[\bA_{h, 1}, \bA_{h, 2}, \bA_{h, 3}]_{P_h}: \bA_{h, 2} \in \calS_{\bM, P_h, \bA_{h, 1}}, \bA_{h, 3} \in  \{-1,1\}^{s_{\le h-1}\times d} \}$ 
\If{$|\calS_{\bM, P_h, \bA_{h, 1}}| \leq \Gamma_h$}
\State $\mA_h \leftarrow \mA_h \cup \mJ_{\bM, P_h, \bA_{h, 1}}$
\EndIf
\EndFor
\EndFor
\State $\mA \leftarrow \mA_1 \cup \cdots \cup \mA_H$
\end{algorithmic}
\end{algorithm}

\subsection{Analysis}

Our goal is to prove the size of $\mA$ is at most $2^{d^2/2}/d$, but at the same time, it contains $\nA$, this would reach a contradiction. In particular, we prove

\begin{lemma}[Upper bound on $\mA$]
\label{lem:upper-size}
The size of $\mA$ is at most $2^{d^2/2}/d$.
\end{lemma}

\begin{lemma}[Lower bound on $\mA$]
\label{lem:lower-size}
$\nA \subseteq \mA$.
\end{lemma}

The proof of Lemma \ref{lem:upper-size} is straightforward and follows simply from counting.
\begin{lemma}
For each level $h\in [H]$, the size of $\mA_h$ is at most $2^{d^2/2}/d^2$. 
\end{lemma}

\begin{proof}
For any fixed $h \in [H]$, the total number of different messages $\bM \in \{0,1\}^{kd}$, partitions $P_h = (P_{h,1}, P_{h,2}, P_{h,3})$ and sub-matrices $\bA_{h, 1} \in  \{-1,1\}^{(d/2-s_{\le h})\times d}$ are at most
\begin{align}
2^{kd} \times 2^{3d} \times 2^{d(d/2 - s_{\le h})} \label{eq:count1}
\end{align}
It remains to bound the size of $\mJ_{\bM, P_h, \bA_{h, 1}}$. Note it would not be counted if $|\calS_{\bM, P_h, \bA_{h, 1}}| > \Gamma_h$.
On the other side, when $|\calS_{\bM, P_h, \bA_{h, 1}}| \leq \Gamma_{h}$, the total different choice of $\bA_{h, 2}$ is at most $\Gamma_h = 2^{ds_h -2kd}$. We make no restrictions on $\bA_{h, 3}$ and its size is at most $2^{d s_{\le h-1}}$. Multiply these numbers, we have
\begin{align}
|\mJ_{\bM, P_h, \bA_{h, 1}}| \leq 2^{ds_h -2kd}\times 2^{d s_{\le h-1}}.\label{eq:count2}
\end{align}
Combining Eq.~\eqref{eq:count1}\eqref{eq:count2}, we have
\begin{align*}
|\mA_h| \leq (2^{kd} \times 2^{3d} \times 2^{d(d/2 - s_{\le h})}) \times (2^{d s_h - 2kd} \times 2^{ds_{\le h-1}}) = 2^{d^2/2 +3d -kd} \leq 2^{d^2/2}/d^2.
\end{align*}
We complete the proof here.
\end{proof}

We set to prove Lemma \ref{lem:lower-size}, which is the major technical Lemma. We set
\[
N_h = 2^{s_h/\log (d)} \quad \forall h \in [0:H].
\]

\begin{proof}[Proof of Lemma \ref{lem:lower-size}]
For any fixed $\bA \in \nA$, we wish to prove $\bA \in \mA$. We prove by induction on $h$, and our inductive hypothesis is

\begin{itemize}
\item {\bf Inductive hypothesis}. If $\bA \notin \mA_{\ell}$ for any $\ell =0,1,\ldots, h$, then there is a submatrix $\wt{\bA}_{h} \subseteq \bA$ ($\wt{\bA}_{h} \in \{-1,1\}^{(d/2 - s_{\le h})\times d})$, such that $|\mN(\mO^{\bM_{\bA}}(\wt{\bA}_{h}), \alpha_{h})|  \geq N_{h}$. 
\end{itemize}
In another word, our inductive hypothesis asserts that if $\bA$ has not yet been added to $\mA$ till level $h$, then there is a submatrix $\wt{\bA}_{h}$ that takes $(d/2 - s_{\le h})$ rows from $\bA$ and the $\alpha_{h}$-cover of $\mO^{\bM_{\bA}}(\wt{\bA}_{h})$ has size at least $N_{h}$.

The base case of $h=0$ holds directly from the definition of $\nA$. In particular, it is clear that $\bA \notin \mA_0 = \emptyset$ and we would take $\wt{\bA}_0 = \bA$. By Lemma \ref{lem:exist}, we have that $|\mN(\mO^{\bM_{\bA}}(\bA), \alpha_{0})|  \geq N_{0}$.

Suppose the claim holds up to $h$, and the orthogonal entries $\mO^{\bM_{\bA}}(\wt{\bA}_h)$ has a large $\alpha_h$-cover.
Our first step is to remove $s_{h+1} = s_h \cdot (\frac{\Delta}{\log^5 (d)})$ different rows, $\bR_{h} \in \{-1,1\}^{s_{h+1}\times d}$, from $\wt{\bA}_h$, and prove that the remaining sub-matrix $\wt{\bA}_{h+1} = \wt{\bA}_h \backslash \bR_h$ still contains a large number of entries from the $\alpha_h$ cover. In particular,

\begin{lemma}[Deleting rows]
\label{lem:delete}
Suppose 
\[
\left|\mN(\mO^{\bM_{\bA}}(\wt{\bA}_{h}), \alpha_{h})\right| \geq N_h = 2^{s_h/\log(d)},
\] 
then there exists $s_{h+1}$ rows of $\wt{\bA}$, denoted as $\bR_h$, such that, after deleting $\bR_h$ from $\wt{\bA}_h$, the remaining matrix $\wt{\bA}_{h+1} = \wt{\bA}_h \backslash \bR_h$ satisfies, 
\begin{align*}
\left|\mO^{\bM_\bA}(\wt{\bA}_{h+1}) \cap \mN(\mO^{\bM_\bA}(\wt{\bA}_{h}), \alpha_{h})\right|  \geq 2^{s_h/2\log(d)}.
\end{align*}
\end{lemma}

\begin{proof}
It is clear if an entry $\bx \in \mN(\mO^{\bM_\bA}(\wt{\bA}_{h}), \alpha_{h})$, then it satisfies 
$
\|\wt{\bA}_{h+1}\bx\|_\infty \leq \|\wt{\bA}_h \bx\|_\infty \leq \xi'.
$
Hence, all we need to prove is that a non-trivial amount of entries in $\mN(\mO^{\bM_\bA}(\wt{\bA}_{h}), \alpha_{h})$ survive after the deletion. Indeed, for any submatrix $\wt{\bA} \subseteq \wt{\bA}_h$ of at least $\frac{d}{10}$ rows, we wish to prove, there exists a row $\br\in \wt{\bA}$, such that
\begin{align}
\left|\mO^{\bM_{\bA}}(\wt{\bA}\backslash \{\br\}) \cap \mN(\mO^{\bM_\bA}(\wt{\bA}_{h}), \alpha_{h})\right| \geq \left(1 -\frac{10}{\Delta}\right)\cdot \left|\mO^{\bM_{\bA}}(\wt{\bA}) \cap \mN(\mO^{\bM_\bA}(\wt{\bA}_{h}), \alpha_{h})\right|.\label{eq:delete-goal}
\end{align}
Assuming this is true, one can take a sequence of rows $\br_1, \ldots, \br_{s_{h+1}}$ and
\begin{align*}
&~ \left|\mO^{\bM_\bA}(\wt{\bA}_{h} \backslash \{\br_1, \ldots, \br_{s_{h+1}}\}) \cap \mN(\mO^{\bM_\bA}(\wt{\bA}_{h}), \alpha_{h})\right| \\
\geq &~ \left(1 -\frac{10}{\Delta}\right)\cdot\left|\mO^{\bM_\bA}(\wt{\bA}_{h} \backslash \{\br_1, \ldots, \br_{s_{h+1}-1}\}) \cap \mN(\mO^{\bM_\bA}(\wt{\bA}_{h}), \alpha_{h})\right| \\
\vdots\\
\geq &~ \left(1 -\frac{10}{\Delta}\right)^{s_{h+1}} \cdot \left|\mO^{\bM_\bA}(\wt{\bA}_{h}) \cap \mN(\mO^{\bM_\bA}(\wt{\bA}_{h}), \alpha_{h})\right|\\
\geq &~ 2^{-s_h/\log^4(d)} \cdot 2^{s_h/\log(d)} \geq  2^{s_h/2\log(d)} 
\end{align*}
The first and second step follows from Eq.~\eqref{eq:delete-goal} and $\wt{\bA}_h$ contains $d/2 - s_{\le h} \geq d/2 - d/5 \geq \frac{d}{5}$ rows.
We plug in the value of $s_{h+1} = s_h \cdot (\frac{\Delta}{\log^5 (d)})$ and $N_{h} = 2^{s_h/\log(d)}$ in the third step.

It remains to prove Eq.~\eqref{eq:delete-goal}, it follows from a counting argument. In particular, we claim
\begin{align*}
&~ \sum_{\br \in \row(\wt{\bA})}\left|\mO^{\bM_\bA}(\wt{\bA}\backslash \{\br\}) \cap \mN(\mO^{\bM_\bA}(\wt{\bA}_{h}), \alpha_{h})\right|\\
\geq &~ (|\row(\wt{\bA})| - (d/\Delta)) \cdot \left|\mO^{\bM_\bA}(\wt{\bA}) \cap \mN(\mO^{\bM_\bA}(\wt{\bA}_{h}), \alpha_{h})\right|.
\end{align*}
This is because for any entry $\bx$ in the RHS (i.e. $\bx \in \mO^{\bM_\bA}(\wt{\bA}) \cap \mN(\mO^{\bM_\bA}(\wt{\bA}_{h}), \alpha_{h})$), suppose it equals $\bx_{\bM_\bA, \{\bq_1, \ldots, \bq_n\}, \bv}$ for some rows $\bq_1, \ldots, \bq_n \in \row(\wt{\bA}) \cup \{\nil\}$ and $\bv\in V^{*}$, then it is also contained in $\mO^{\bM_\bA}(\wt{\bA}\backslash \{\br\}) \cap \mN(\mO^{\bM_\bA}(\wt{\bA}_{h}), \alpha_{h})$ as long as $\br\notin \{\bq_1, \ldots, \bq_n\}$. This means $\bx$ has been counted for $|\row(\wt{\bA}) - n| = |\row(\wt{\bA}) - (d/\Delta)|$ times in the LHS.

Combining with the assumption that $\wt{\bA}$ contains at least $d/10$ rows, we conclude the proof of Eq.~\eqref{eq:delete-goal}, and complete the proof of Lemma.
\end{proof}


We continue the proof of induction. In particular, let $\bR_{h}, \wt{\bA}_{h+1}$ be defined as in Lemma \ref{lem:delete}, we would consider the following enumeration
\begin{align}
\bM = \bM_\bA, \quad \bA_{h+1, 1} = \wt{\bA}_{h+1} \in \{-1,1\}^{(d/2 - s_{\le h+1}) \times d},  \quad \bA_{h+1, 3} = \bA\backslash \wt{\bA}_{h} \in \{-1,1\}^{s_{\le h} \times d} \label{eq:enum}
\end{align}
and the partition $P_{h+1} \in \sP_{h+1}$ is the one, such that, $\bA = [\wt{\bA}_{h+1}, \bR_h, \bA\backslash \wt{\bA}_{h}]_{P_{h+1}}$.

It remains to fill in $\bA_{h+1, 2}$ with $\bR_{h}$, to this end, we prove
\begin{lemma}
\label{lem:exist}
Suppose the message $\bM$, the matrix $\bA_{h+1,1}$ and the partition $P_{h+1}$ are defined as Eq.~\eqref{eq:enum}, then $\bR_h \in \calS_{\bM, P_{h+1}, \bA_{h+1,1}}$
\end{lemma}
\begin{proof}
Consider the set 
\begin{align}
X := \mO^{\bM}(\bA_{h+1,1}) \cap \mN(\mO^{\bM}(\wt{\bA}_{h}), \alpha_{h})  \label{eq:def-exist-x}
\end{align}
By Lemma \ref{lem:delete}, the size $|X| \geq 2^{s_h/2\log(d)}$ and for any $\bx \in X$, $\|\bR_h\bx\|_\infty \leq \|\wt{\bA}_h\bx\|_\infty \leq \xi'$. Therefore, it suffices to prove that there exists an $(\alpha_{h}/4)$-RLI sequence $\bx_1, \ldots, \bx_L \in X$, note it would imply $\bR_h \in \mI_{h+1, \bx_1,\ldots, \bx_L}$.

We prove by contradiction and assume the maximum length of $(\alpha_{h}/4)$-RLI sequence is at most $L' < L$, then there exists an $L'$-dimensional subspace (formed by the maximum sequence), represented by the orthonormal $\bU = [\bu_1, \ldots \bu_{L'}] \in \R^{d\times L'}$, such that 
\begin{align}
\|(\bI - \bU\bU^\top)\bx\|_2 < \frac{\alpha_h}{4}, \quad  \forall \bx\in X. \label{eq:exist1}
\end{align}

Consider the grid 
\[
G_{L'} = \left\{ \sum_{i=1}^{L'} \lambda_i \bu_i: \lambda_i = 0 \pm \xi, \pm 2\xi, \ldots, \pm 1, \forall i \in [L']   \right\} 
\]
Then the size of the grid is most
\[
|G_{L'}| \leq (4/\xi)^{L'} \leq \exp(L\log(4/\xi)) <  \exp\left(s/2\log^3(d)\right) \leq \exp(s_h/2\log(d)).
\]
where the third step follows from the choice of parameters $L = \frac{s}{4\log^8(d)}$, $\xi = 2\exp(-\log^5(d))$, the last step follows from $s_h \geq s_0 = s/\log^2(d)$.

Therefore, there exists $\bx_1, \bx_2 \in X$ such that the projections $\bU\bU^\top \bx_1, \bU\bU^\top \bx_2$ are in the same grid, then we have
\begin{align*}
\|\bx_1 - \bx_2\|_2 \leq \|(\bU\bU^\top)(\bx_1 - \bx_2)\|_2 + \|(\bI - \bU\bU^\top)(\bx_1 - \bx_2)\|_2 \leq 2d\xi + 2\cdot \frac{\alpha_h}{4} < \alpha_h
\end{align*}
Here the first step follows from the triangle inequality, the second step follows from Eq.~\eqref{eq:exist1} and $\bU\bU^\top\bx_1, \bU\bU^\top \bx_2$ are at the same grid. This contradicts with the fact that points in $X$ have pairwise distance at least $\alpha_h$ (see the definition at Eq.~\eqref{eq:def-exist-x}). We complete the proof here.
\end{proof}

Due to Lemma \ref{lem:exist}, we can conclude that, if $|\calS_{\bM, P_{h+1}, \bA_{h+1,1}}| \leq \Gamma_{h+1} = 2^{s_{h+1} d -2kd}$, then the encoding algorithm would add $\bA$ to $\mA_{h+1}$, and we can finish the induction. Hence, we focus on the case of $|\calS_{\bM, P_{h+1}, \bA_{h+1,1}}| > \Gamma_{h+1} = 2^{s_{h+1} d -2kd}$. 
In this case, let $\hat{N}_{h+1}$ be the size of $\alpha_{h+1}$-cover in $\mO^{\bM}(\bA_{h+1,1})$, that is
\[
\hat{N}_{h+1} := |\mN(\mO^{\bM}(\bA_{h+1,1}), \alpha_{h+1})|. 
\]

We give a lower bound on $\hat{N}_{h+1}$ in terms of $|\calS_{\bM, P_{h+1}, \bA_{h+1,1}}|$, in particular, we prove
\begin{lemma}
\label{lem:size-s}
Suppose the message $\bM$, the matrix $\bA_{h+1,1}$ and the partition $P_{h+1}$ are defined as Eq.~\eqref{eq:enum}, then 
$|\calS_{\bM, P_{h+1}, \bA_{h+1,1}}| \leq (\hat{N}_{h+1})^{L} \cdot 2^{s_{h+1} \cdot (d - \Omega(L))}$.
\end{lemma}




We need the following bound on the number of orthogonal vectors to a RLI sequence. 

\begin{lemma}[Number of orthogonal vectors]
\label{lem:num-row}
Given any $(\alpha_{h}/4)$-RLI sequence $\bx_{1}, \ldots, \bx_{L}$, define
\begin{align*}
\wt{\mI}_{h+1, \bx_1, \ldots, \bx_{L}} := \bigcup_{\wt{\bx}_i \in \mB(\bx_i, 2\alpha_{h+1}), \forall i \in [L]} \mI_{h+1, \wt{\bx}_1, \ldots, \wt{\bx}_{L}}.
\end{align*}
Then we have 
\begin{align*}
\left|\wt{\mI}_{h+1, \bx_1, \ldots, \bx_{L}}\right| \leq 2^{s_{h+1} (d - c L)}
\end{align*}
for some absolute constant $c > 0$.
\end{lemma}
\begin{proof}
We apply Lemma \ref{lem:rli} for the sequence $\bx_1, \ldots, \bx_{L}$,  with $\delta = (\alpha_h/4)^{2}/2$, $\bX = [\bx_1, \ldots, \bx_{L}]\in \R^{d\times L}$, then there exists an orthonormal matrix $\bU \in \R^{d\times (L/2)}$, such that
\begin{align}
\label{eq:num-row1}
\|\bU^\top \ba\|_\infty \leq \frac{d}{\delta}\|\bX^\top\ba\|_\infty = \frac{32d}{\alpha_h^2} \|\bX^\top\ba\|_\infty \quad \forall \ba \in \R^d.
\end{align}
Now for any matrix $\bR \in \wt{\mI}_{h+1, \bx_1, \ldots, \bx_{L}}$, there exists a sequence $\wt{\bx}_{1}, \ldots, \wt{\bx}_{L}$, such that $\wt{\bx}_i \in \mB(\bx_i, 2\alpha_{h+1})$ for $i\in [L]$, and $\bR \in \mI_{h+1, \wt{\bx}_1, \ldots, \wt{\bx}_{L}}$. 
Let $\wt{\bX} = [\wt{\bx}_1, \ldots, \wt{\bx}_{L}] \in \R^{d\times L}$, for each row $\br$ of $\bR$, one has
\begin{align}
\|\bU^\top \br\|_\infty \leq \frac{32d}{\alpha_h^2}\|\bX^\top \br\|_\infty \leq \frac{32d}{\alpha_h^2} (\|\wt{\bX}^\top\br\|_\infty + \sqrt{d}\cdot 2\alpha_{h+1}) \leq \frac{32d\xi'}{\alpha_h^2} + \frac{64d\sqrt{d}\alpha_{h+1}}{\alpha_h^2} \leq 1/d^5. \label{eq:num-row2}
\end{align}
Here the first step follows from Eq.~\eqref{eq:num-row1}, the second step follows from $\wt{\bx}_i \in \mB(\bx_i, 2\alpha_{h+1})$ for every $i \in [L]$ and $\|\br\|_2 = \sqrt{d}$. The last step follows from the choice of parameters.

Let $\mR_{\mathsf{or}} \subseteq \{-1,1\}^{d}$ contain all possible $\br$ that satisfies Eq.~\eqref{eq:num-row2}.
We wish to bound the size of $\mR_{\mathsf{or}}$, by Lemma \ref{lem:probabilistic-v}, we have that
\[
\Pr_{\br\sim \{-1,1\}^{d}}\left[\|\bU^\top \br\|_2^2 \leq 1/d^8\right] \leq 2^{-cL}
\]
for some constant $c > 0$. 
Hence, we have $|\mR_{\mathsf{or}}| \leq 2^{d}\cdot 2^{-cL} = 2^{d-cL}$. Since we have proved $\br \in \mR_{\mathsf{or}}$ for every row of $\bR \in \wt{\mI}_{h+1, \bx_1, \ldots, \bx_L}$, and $\bR$ has $s_{h+1}$ rows, we conclude the proof of Lemma.
\end{proof}

Now we can go back to the proof of Lemma \ref{lem:size-s}
\begin{proof}[Proof of Lemma \ref{lem:size-s}]
Let $X_{h}$ contain all $(\alpha_h/4)$-RLI sequence of length $L$ in $\mO^{\bM}(\bA_{h+1, 1})$, then ones has
\[
\calS_{\bM, P_{h+1}, \bA_{h+1,1}} = \bigcup_{(\bx_1, \ldots, \bx_{L}) \in X_{h}} \mI_{h+1, \bx_1, \ldots, \bx_{L}}.
\]

Construct $Y_h$ as follows. For any $\bx_1, \ldots, \bx_{L}$ in the $\alpha_{h+1}$-cover $\mN(\mO^{\bM}(\bA_{h+1,1}), \alpha_{h+1})$, if there exists a sequence $\by_1, \ldots, \by_{L}$, such that 
\begin{itemize}
\item $\by_i \in \mB(\bx_{i}, \alpha_{h+1})$,
\item $\by_i \in \mO^{\bM}(\bA_{h+1, 1})$, and 
\item $\by_1, \ldots, \by_{L}$ forms an $(\alpha_h/4)$-RLI sequence,
\end{itemize}
then we add $(\by_1, \ldots, \by_{L})$ to $Y_{h}$. If there are multiple sequences, we only add one of them.

First, we have $|Y_{h}| \leq (\hat{N}_{h+1})^{L}$ because we look at each $L$-tuple in $\mN(\mO^{\bM}(\bA_{h+1,1}), \alpha_{h+1})$ once. 

Next, we claim that for any sequence $\bx_1, \ldots, \bx_{L} \in X_h$, there exists $(\by_1, \ldots, \by_L) \in Y_{h}$ such that $\|\bx_i - \by_i\|_2 \leq 2\alpha_{h+1}$. To see this, let $\wt{\bx}_1, \ldots, \wt{\bx}_L$ from the $\alpha_{h+1}$-cover $\mN(\mO^{\bM}(\bA_{h+1,1}), \alpha_{h+1})$, such that $\|\wt{\bx}_i - \bx_i\|_2 \leq \alpha_{h+1}$. Then there must exists $(\by_1, \ldots, \by_L)\in Y_h$ such that $\|\wt{\bx}_i - \by_i\|_2 \leq \alpha_{h+1}$. This is sufficient for our purpose.

Hence, we can conclude that
\[
\calS_{\bM, P_{h+1}, \bA_{h+1,1}} = \bigcup_{(\bx_1, \ldots, \bx_{L}) \in X_{h}} \mI_{h+1, \bx_1, \ldots, \bx_{L}} \subseteq  \bigcup_{(\by_1, \ldots, \by_{L}) \in Y_{h}} \wt{\mI}_{h+1, \by_1, \ldots, \by_{L}} 
\]
By Lemma \ref{lem:num-row}, we conclude
\[
|\calS_{\bM, P_{h+1}, \bA_{h+1,1}}|\leq |Y_h| \cdot 2^{s_{h+1}(d-cL)} \leq (\hat{N}_{h+1})^{L} \cdot 2^{s_{h+1} \cdot (d - \Omega(L))}.
\]
This completes the proof of Lemma.
\end{proof}

By Lemma \ref{lem:size-s} and our assumption that $|\calS_{\bM, P_{h+1}, \bA_{h+1,1}}| \geq \Gamma_{h+1}$, we have that
\begin{align*}
(\hat{N}_{h+1})^{L}\cdot 2^{s_{h+1} \cdot (d - \Omega(L))} \geq |\calS_{\bM, P_{h+1}, \bA_{h+1,1}}| \geq \Gamma_{h+1} = 2^{s_{h+1}d-2kd} 
\end{align*}
and therefore,
\begin{align*}
\hat{N}_{h+1}\geq 2^{s_{h+1}/\log(d)} = N_{h+1}
\end{align*}
this is because $kd = d^{2-\delta-o(1)}$ and $s_{h+1} L \geq s^2/\poly\log(d) \geq d^{2-\delta}/\poly\log(d)$.

This wraps up the induction.

For any matrix $\mA \in \nA$, if $\mA \notin \mA_{h}$ for any $h=0,1,2,\ldots, H$, then the above induction implies that there exists a sub-matrix $\wt{\bA}_{H} \subseteq \bA$ ($\wt{\bA}_H \in \{-1,1\}^{(d/2 - s_{\le H})\times d}$), such that
\[
|\mN(\mO^{\bM_\bA}(\wt{\bA}_{H}), \alpha_{H})|  \geq N_{H} = 2^{s_H/\log d} \geq 2^{d/10\log(d)}.
\]
This is simply impossible, because the total number of entries in $\mN(\mO^{\bM_\bA}(\wt{\bA}_{H}), \alpha_{H})$ is at most 
\begin{align*}
\left|\mN(\mO^{\bM_\bA}(\wt{\bA}_{H}), \alpha_{H})\right| \leq  &~ |\T^{\bM_\bA}(\bA)| =  |\T^{\bM_\bA}(\bA, V^{*})|\\
\leq &~ |V^{*}| \cdot (d/2+1)^{n} \leq 2^{s + n\log_2(d)} \leq 2^{d/10\log(d)}.
\end{align*}
Here the third follows from the definition of a table (see Definition \ref{def:table}), and there are at most $(d/2+1)^{n}$ combinations of the third round message given $\bA$, the fourth step follows from $|V^*|\leq 2^s$ (see Lemma \ref{lem:exist-v}) -- this is the only place that we use the size of $V^{*}$ is not large.

Combining the above inequalities, we have proved that $\nA \subseteq \mA$ and finish the proof of Lemma \ref{lem:lower-size}.
\end{proof}

\section*{Acknowledgement}
X.C. and B.P. wish to thank Daniel Hsu, Christos Papadimitriou, Aviad Rubinstein, Gregory Valiant for helpful discussion over the project.

\bibliography{ref}
\bibliographystyle{alpha}

\newpage
\appendix
\section{Useful Lemma}

\begin{lemma}[Khintchine’s Inequality]
\label{lem:khintchine}
Let $\sigma_1, \ldots, \sigma_{n}$ be i.i.d. Rademacher variables (i.e., $\sigma_i \sim \{-1,1\}$), and let $x_1, \ldots, x_n$ be real numbers. Then there are constants $c_1, c_2 > 0$ so that
\begin{align*}
\Pr\left[\left|\sum_{i=1}^{n}\sigma_i x_i\right| \geq c_1 t\cdot \|x\|_2\right] \leq \exp(-c_2 t^2).
\end{align*}
\end{lemma}

Define the sub-Gaussian norm $\|x\|_{\psi_2}$ of a sub-Gaussian random variable $x$ as
\[
\|x\|_{\psi_2} := \inf \{ K > 0 \text{ such that } \E[\exp(x^2/K^2)] \leq 2\}.
\]
We have
\begin{lemma}[Projection of sub-gaussian random variables, Lemma 40 of \cite{marsden2022efficient}]
Let $\bx\in \R^d$ be a random vector with i.i.d sub-Gaussian components  which satisfy $\E[\bx_i]= 0, \|\bx_i\|_{\psi_2} \leq K$, and $\E[\bx\bx^\top]$ = $s^2 I_d$. Let $\bU \in \R^{d\times r}$ be an orthonormal matrix, then there is a constant $c > 0$, such that, for any $t\geq 0$
\begin{align*}
\Pr\left[|\|\bU^\top \bx\|_2^2 - rs^2| \geq t \right] \leq \exp\left(-c\min\left\{\frac{t^2}{rK^4}, \frac{t}{K^2}   \right\}\right).
\end{align*}
\end{lemma}

In particular, if the vector $\bv \sim \frac{1}{\sqrt{d}}\mH_d$, then we have $\E[\bv\bv^\top] = \frac{1}{d} I$ and $\|\bv_i\|_{\psi_2} \leq \frac{2}{\sqrt{d}}$, and there exists a constant $c > 0$ such that for any $t\geq 0$,
\begin{lemma}[Projection of random vectors in $\mH_d$]
\label{lem:projection-prob}
Let $\bv \sim \frac{1}{\sqrt{d}}\mH_d$ and $\bU \in \R^{d\times r}$ be an orthonormal matrix, then
\begin{align*}
\Pr\left[\left|\|\bU^\top \bv\|_2^2 - \frac{r}{d}\right| \geq t \right] \leq \exp\left(-c\min\left\{\frac{d^2t^2}{16r}, \frac{dt}{4}   \right\}\right).
\end{align*}
\end{lemma}

\begin{lemma}[Lemma 34 from \cite{marsden2022efficient}]
\label{lem:rli}
Let $L \in [d]$, $\delta \in (0, 1]$.
Suppose a sequence of unit norm vectors $\bx_1, \ldots,\bx_{L} \in \R^d$ satisfies
\[
\|\proj_{\mathsf{span}(\bx_1, \ldots, \bx_{i-1})}(\bx_i)\|_2 \leq 1-\delta.
\]
Let $\bX = [\bx_1, \ldots, \bx_{L}] \in \R^{d\times L}$. There exists an orthonormal matrix $\bU \in \R^{d\times (L/2)}$ such that for any vector $\ba\in \R^d$,
\[
\|\bU^\top \ba\|_\infty \leq \frac{d}{\delta}\|\bX^\top \ba\|_\infty
\]
\end{lemma}

\end{document}